\documentclass[11pt]{article}

\usepackage{etoolbox}

\usepackage[utf8]{inputenc}

\usepackage[letterpaper,margin=1in]{geometry}

\usepackage{amsmath, amssymb}

\usepackage{amsthm}

\usepackage{complexity}

\usepackage{dsfont}
\usepackage{bbm}

\usepackage{hyperref}
\hypersetup{colorlinks, linkcolor=darkblue, citecolor=darkgreen, urlcolor=darkblue}
\usepackage[capitalize, nameinlink]{cleveref}
\crefname{theorem}{Theorem}{Theorems}
\Crefname{lemma}{Lemma}{Lemmas}
\Crefname{claim}{Claim}{Claims}
\Crefname{fact}{Fact}{Facts}
\Crefname{remark}{Remark}{Remarks}
\Crefname{observation}{Observation}{Observations}
\Crefname{figure}{Figure}{Figures}
\Crefname{line}{Line}{Lines}
\Crefname{algocf}{Algorithm}{Algorithms}
\crefalias{AlgoLine}{line}
\Crefname{stepsromani}{Step}{Steps}
\Crefname{stepsarabici}{Step}{Steps}

\newtheorem{theorem}{Theorem}
\newtheorem{lemma}[theorem]{Lemma}

\newtheorem{definition}[theorem]{Definition}
\newtheorem{remark}[theorem]{Remark}
\newtheorem{corollary}[theorem]{Corollary}
\newtheorem{assumption}[theorem]{Assumption}

\newtheorem{claim}[theorem]{Claim}

\usepackage{times}

\usepackage{anyfontsize}

\usepackage[inline]{enumitem}
\setlist[enumerate,1]{label=(\roman*), leftmargin=2.2em}
\setlist[enumerate,2]{label=(\alph*)}
\setlist{nosep,topsep=0.1em}
\setlist[itemize,1]{label={\bfseries--}}

\newlist{stepsroman}{enumerate}{1}
\setlist[stepsroman]{label=(\roman*), leftmargin=2.2em}

\newlist{stepsarabic}{enumerate}{1}
\setlist[stepsarabic]{rightmargin=0.2em, label=\arabic*., ref=\arabic*}

\usepackage{mathtools}

\usepackage{mdframed}
\mdfsetup{skipabove=\bigskipamount, skipbelow=\bigskipamount, innerleftmargin=0.5em, innertopmargin=0.5em, innerrightmargin=0.5em, innerbottommargin=0.5em, align=center}

\usepackage[color=darkgreen!40!white]{todonotes}
\setuptodonotes{size=\scriptsize}

\usepackage{xcolor}
\definecolor{darkblue}{rgb}{0,0,0.38}
\definecolor{darkred}{rgb}{0.6,0,0}
\definecolor{darkgreen}{rgb}{0.1,0.35,0}

\usepackage[
backend=biber,
style=alphabetic,
citestyle=alphabetic,
maxalphanames=6,
maxcitenames=99,
mincitenames=98,
maxbibnames=99,
giveninits=true,
url=false,
doi=true,
isbn=true,
backref=true
]{biblatex}

\usepackage{xpatch}

\makeatletter
\patchcmd\blx@bblinput{\blx@blxinit}
                      {\blx@blxinit
                      }{}{\fail}
\makeatother

\newcommand{\footremember}[2]{\footnote{#2}
    \newcounter{#1}
    \setcounter{#1}{\value{footnote}}}
\newcommand{\footrecall}[1]{\footnotemark[\value{#1}]}

\makeatletter
 \newcommand{\linkdest}[1]{\Hy@raisedlink{\hypertarget{#1}{}}}
\makeatother
\newcommand{\TSP}{\protect\hyperlink{prb:TSP}{TSP}}
\newcommand{\PCTSP}{\protect\hyperlink{prb:PCTSP}{PCTSP}}
\newcommand{\PCST}{\protect\hyperlink{prb:PCST}{PCST}}

\newcommand{\Exp}{\mathbb{E}}
\newcommand{\Prob}{\mathbb{P}}
\newcommand{\odd}{\operatorname{odd}}

\makeatletter
\DeclareFieldFormat{eprint:arxiv}{arXiv\addcolon\space
  \ifhyperref
    {\href{https://arxiv.org/abs/#1}{\nolinkurl{#1}\iffieldundef{eprintclass}
     {}
{\addspace\mkbibbrackets{\thefield{eprintclass}}}}}
    {\nolinkurl{#1}
     \iffieldundef{eprintclass}
       {}
{\addspace\mkbibbrackets{\thefield{eprintclass}}}}}
\makeatother

\makeatletter
\newcommand{\labeltarget}[1]{\Hy@raisedlink{\hypertarget{#1}{}}}
\makeatother

\usepackage{wrapfig}

\usepackage[margin=10pt,font=small,labelfont=bf]{caption}

\usepackage{subcaption}

\usepackage{graphicx}
\graphicspath{{.}{graphics/}}
\makeatletter
\newcommand\appendtographicspath[1]{\g@addto@macro\Ginput@path{#1}}
\makeatother

\usepackage{tikz}

\usepackage[vlined,ruled,algo2e]{algorithm2e}
\SetAlCapHSkip{0.2em}
\SetAlgoInsideSkip{smallskip}
\SetCustomAlgoRuledWidth{0.95\textwidth}
\makeatletter
\patchcmd{\@algocf@start}{\begin{lrbox}{\algocf@algobox}}{\rule{0.025\textwidth}{\z@}\begin{lrbox}{\algocf@algobox}\begin{minipage}{0.95\textwidth}}{}{}
\patchcmd{\@algocf@finish}{\end{lrbox}}{\end{minipage}\end{lrbox}}{}{}
\makeatother

\usepackage{xfrac}
\ExplSyntaxOn
\DeclareRestrictedTemplate { xfrac } { text } { math }
  {
    numerator-font      = \number \fam ,
    slash-symbol        = /            ,
    slash-symbol-font   = \number \fam ,
    denominator-font    = \number \fam ,
    scale-factor        = 0.7          ,
    scale-relative      = false        ,
    scaling             = true         ,
    denominator-bot-sep = 0 pt         ,
    math-mode           = true         ,
    phantom             = ( }
\DeclareInstance { xfrac } { mathdefault } { math }
  { numerator-top-sep = 0pt }
\ExplSyntaxOff

\newcommand{\bb}[2][V_1]{\operatorname{backbone}_{#1}(#2)}
\newcommand{\lb}[2][V_1]{\operatorname{limbs}_{#1}(#2)}
\newcommand{\dd}{\operatorname{d}\!}
\newcommand{\PHK}{P_{\text{HK}}}
\newcommand{\PST}{P_{\text{ST}}}

\newcommand{\walks}{{\mathcal{W}_0}}

\newcommand{\spare}[1]{\text{spare}_{#1}}

\tikzstyle{Vertex}=[circle,fill=black,minimum size=6,inner sep=0pt]
\tikzstyle{Vertex2}=[draw=black,circle,fill=white,minimum size=4,inner sep=0pt]
\title{An improved approximation guarantee for Prize-Collecting TSP}

\author{Jannis Blauth\footremember{UBonnHCM}{Research Institute for Discrete Mathematics and Hausdorff Center for Mathematics, University of Bonn, Bonn, Germany.
Email:
\href{mailto:blauth@or.uni-bonn.de}{blauth@or.uni-bonn.de},
\href{mailto:mnaegele@uni-bonn.de}{mnaegele@uni-bonn.de}. The second author is supported by the Swiss National Science Foundation (grant no.\ P500PT\_206742) and the Deutsche Forschungsgemeinschaft (DFG, German Research Foundation) under Germany's Excellence Strategy~--~EXZ-2047/1~--~390685813.
}\and
Martin Nägele\footrecall{UBonnHCM}
}\date{}

\begin{document}

\maketitle

\begin{abstract}
We present a new approximation algorithm for the (metric) prize-collecting traveling salesperson problem (PCTSP).
In PCTSP, opposed to the classical traveling salesperson problem (TSP), one may	choose to not include a vertex of the input graph in the returned tour at the cost of a given vertex-dependent penalty, and the objective is to balance the length of the tour and the incurred penalties for omitted vertices by minimizing the sum of the two.
We present an algorithm that achieves an approximation guarantee of $1.774$ with respect to the natural linear programming relaxation of the problem.
This significantly reduces the gap between the approximability of classical TSP and PCTSP, beating the previously best known approximation factor of $1.915$.
As a key ingredient of our improvement, we present a refined decomposition technique for solutions of the LP relaxation, and show how to leverage components of that decomposition as building blocks for our tours. \end{abstract}

\bigskip

\thispagestyle{empty}
\addtocounter{page}{-1}

\begin{tikzpicture}[overlay, remember picture, shift = {(current page.south east)}]
\node[anchor=south east, outer sep=5mm] {
\begin{tikzpicture}[outer sep=0] \node (dfg) {\includegraphics[height=12mm]{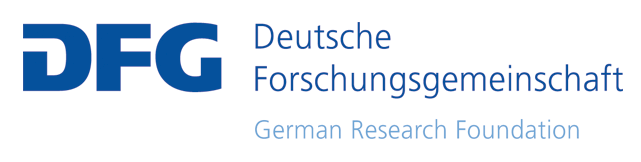}};
\node[left=5mm of dfg, yshift=1mm] (snf) {\includegraphics[height=8mm]{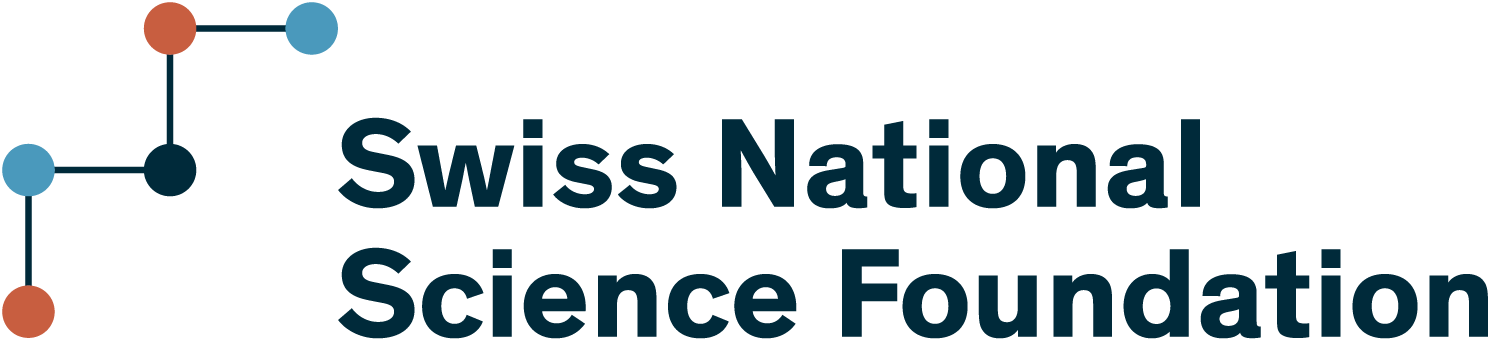}};
\end{tikzpicture}
};
\end{tikzpicture}

\newpage

\section{Introduction}

The classical (metric) traveling salesperson problem (\TSP{}\linkdest{prb:TSP}) is one of the most well-known problems in Combinatorial Optimization.
Given a complete undirected graph $G=(V,E)$ with metric edge lengths $c_e\in\mathbb{R}_{\geq 0}$ for all $e\in E$, the task is to find a shortest Hamiltonian cycle in $G$.
\TSP{} arises canonically in many applied settings, and it is thus not surprising that a vast number of variations of the classical problem are studied both from a theoretical and an applied perspective.
One very natural such variation that we focus on in this paper is \emph{prize-collecting TSP}.
\begin{mdframed}[userdefinedwidth=0.95\linewidth]
{\textbf{Prize-Collecting TSP (\PCTSP{}\linkdest{prb:PCTSP}):}}
Given a complete undirected graph $G=(V,E)$ with metric edge lengths $c_e\geq 0$ for all $e\in E$, a root $r \in V$, and penalties $\pi_v\geq 0$ for all $v\in V \setminus \{r\}$, the task is to find a cycle $C=(V_C, E_C)$ in $G$ that contains the root $r$ and minimizes
$$\sum_{e\in E_C} c_{e} + \sum_{v\in V\setminus V_C} \pi_v\enspace.$$
\end{mdframed}
In other words, for each vertex other than the root, we may decide between covering the vertex by the output cycle or dropping it at cost $\pi_v$, and the goal is to find a balance between the tour cost and incurred penalties so as to minimize the sum of the two.
We remark that this may lead to degenerate cycles $C$ that may contain only two vertices (or even the root only).
Many real-world applications of \TSP{} variants (e.g., in logistics) share the property that it is allowed to not visit some of the vertices at certain extra costs, and \PCTSP{} is one very natural way of modeling such degrees of freedom.

Clearly, \PCTSP{} is at least as hard as \TSP{}: Large enough penalties enforce visiting all vertices.
As a consequence, \PCTSP{} is \APX-hard, and one research focus lies on approximation algorithms.
For \TSP{}, despite extensive research, the $\sfrac32$-approximation algorithm of \textcite{christofides_1976_worst,serdyukov_1978_onekotorykh} (also see \cite{christofides_1976_worst_new,vanBevern_2020_historical}) persisted to be the best known for more than four decades, until a very recent breakthrough of \textcite{karlin_2021_slightly, karlin_derand}, who improved the factor to $\sfrac32-\varepsilon$ for some $\varepsilon>10^{-36}$.
For \PCTSP{}, the first constant factor approximation was given by \textcite{bienstock_1993_note} through a threshold rounding approach based on the following linear programming relaxation, where we use $\pi_r \coloneqq 0$ for convenience:\footnote{For $S\subseteq V$, we let $\delta(S)\coloneqq\{e\in E\colon |e\cap S|=|e\setminus S|=1\}$. Moreover, for $v\in V$, we use $\delta(v)\coloneqq \delta(\{v\})$.}
\begin{equation}\tag{PCTSP LP relaxation}\label{eq:rel_PCTSP}
\begin{array}{rrcll}
\min & \displaystyle \sum_{e\in E}c_ex_e + \sum_{v\in V}\pi_v (1-y_v) \\
     & x(\delta(v)) & = & 2y_v & \forall v\in V\setminus\{r\} \\
     & x(\delta(r)) & \le & 2 \\
     & x(\delta(S)) & \geq & 2y_v & \forall S\subseteq V\setminus\{r\}, v\in S\\
     & y_r & = & 1 \\
     & x_e & \geq & 0 & \forall e\in E\\
     & y_v & \in & [0,1] & \forall v\in V\enspace.
\end{array}
\end{equation}
Here, the variables $y_v$ can be interpreted to indicate the extent to which $v$ is covered by a fractional solution.
Note that we are using a version of the linear relaxation that includes the constraint $x(\delta(r))\leq 2$.
This does not change the optimal value, but, if $|V|\geq 2$, it guarantees that this \ref{eq:rel_PCTSP} generalizes the classical Held-Karp relaxation for \TSP{}, which is recovered by setting $y_v=1$ for all $v\in V$ (corresponding to the requirement that all vertices need to be fully connected to the root).

The threshold rounding algorithm of \cite{bienstock_1993_note} considers an optimal solution of the \ref{eq:rel_PCTSP}, and exploits the Christofides-Serdyukov algorithm to construct a tour on precisely those $v\in V$ with $y_v\geq\sfrac35$.
This tour is shown to have cost within a factor $\sfrac52$ of the value of the optimal LP solution, and thus also of the optimal tour.
Relying on an equivalent LP, \textcite{goemans_1995_general} later gave a $2$-approximation through a primal-dual algorithm.
By exploiting that the latter algorithm in fact loses different factors on the two parts of the objective, \textcite{archer_2011_improved} could break the barrier of $2$ and give a $\sfrac{97}{49}$-approximation algorithm.
Finally, combining a randomized analysis of the threshold rounding approach with the improved analysis of the primal-dual method, \textcite{goemans_2009_combining} obtained the currently best known guarantee of $\sfrac1{\left(1-\frac23e^{-\sfrac13}\right)}\approx 1.915$.
It is worth noting that the latter guarantee is still relative to the lower bound given by the \ref{eq:rel_PCTSP}, and it allows using LP-relative \TSP{} algorithms in a black-box way:
Given an algorithm for \TSP{} that outputs a $\beta$-approximate solution with respect to the standard Held-Karp relaxation, the approach of \textcite{goemans_2009_combining} returns a solution of cost at most $\sfrac1{(1-\frac1\beta e^{1-\sfrac2\beta})}$ times the cost of an optimum solution to the \ref{eq:rel_PCTSP}.

\subsection{Our results and contributions}

Our main result is to improve the approximation guarantee for \PCTSP{}, getting significantly closer to the best known approximation guarantee for \TSP{}.

\begin{theorem}\label{thm:improvedApproximation}
There is an LP-relative $1.774$-approximation algorithm for \PCTSP{}.
\end{theorem}

Our algorithm is similar in spirit to the classical threshold rounding approach that was leveraged in \cite{bienstock_1993_note} and \cite{goemans_2009_combining}, but overcomes one of its critical weaknesses:
In threshold rounding, the two decisions of
\begin{enumerate*}
\item which vertices to visit, and
\item how to build the tour
\end{enumerate*}
are taken independently.
Concretely, based on an optimal solution $(x^*,y^*)$ of the \ref{eq:rel_PCTSP} and a threshold $\gamma\in(0,1]$, the set $V_{\gamma}\coloneqq \{v\in V\colon y_v^*\geq \gamma\}$ of vertices that will be visited by the returned tour is fixed, and any LP-relative \TSP{} algorithm can be used to blindly build a tour on $V_\gamma$.
In our algorithm, to the contrary, choosing which vertices to visit and the actual construction of the tour are carefully intertwined.

Similarly to the threshold rounding algorithm, we also maintain a core set $V_\gamma$ of vertices that are guaranteed to be on the output tour.
Instead of building a tour from edges in $\binom{V_\gamma}{2}$, though, the building blocks of our tour are given by a set $\mathcal{W}$ of walks on $V$ having precisely their endpoints in $V_\gamma$ and that are pre-computed based on the solution $(x^*, y^*)$ of the \ref{eq:rel_PCTSP}.
Through a procedure based on randomized pipage rounding \cite{ageev_2004_pipage,srinivasan_2011_distributions, calinescu_2011_maximizing}, we sample a subset of these walks from $\mathcal{W}$ such that the resulting multigraph $H$ spans all vertices in $V_\gamma$, and its cost can be bounded with respect to the LP solution.
Crucially, through our construction, the graph $H$ will have odd degrees in $V_\gamma$ only, hence the cost of constructing an actual tour based on $H$ (typically called \emph{parity correction}) is not more than in the classical threshold rounding approach.
In comparison, the primal-dual approach by \textcite{goemans_1995_general} constructs a tree without taking care of the vertex parities and ends up with doubling the computed tree for parity correction.
Finally, our randomized tour construction allows us to analyze probabilities for vertices in $V\setminus V_\gamma$ to appear on the tour, and thus get bounds on expected penalties incurred for such vertices.
We also note that derandomizing our construction is possible using a  deterministic variant of pipage rounding.

Due to the refined tour construction, our approach as sketched above does not allow to use any LP-relative approximation algorithm for \TSP{} as a black box. We remark, though, that even if an LP-relative $\sfrac43$-approximation for \TSP{} (matching the current best lower bound on the integrality gap of the Held-Karp relaxation) was available, our approximation guarantee remains superior to the current best black-box algorithm by Goemans~\cite{goemans_2009_combining}. Furthermore, as mentioned in \cref{thm:improvedApproximation}, our analysis is still with respect to the optimal value of the \ref{eq:rel_PCTSP}, thereby immediately giving the subsequent corollary.

\begin{corollary}\label{cor:intGap}
The integrality gap of the \ref{eq:rel_PCTSP} is less than $1.774$.
\end{corollary}

\pagebreak

On the technical side, the main ingredient of our algorithm is the construction and composition of the set $\mathcal{W}$ of walks mentioned earlier.
We obtain appropriate walks through a decomposition of the $x$-part of a solution $(x,y)$ of the \ref{eq:rel_PCTSP}.
Our decomposition of $x\in\mathbb{R}^E$ can equivalently be seen as a packing result on the multigraph obtained through scaling up $x$ to an integral vector, and interpreting $x_e$ as an edge multiplicity for every $e\in E$.
Undoubtedly, the most prominent such packing theorem is Edmond's result on packing disjoint spanning arborescences \cite{edmonds_1973_edge}:
A digraph $D=(V,A)$ admits a packing of $k$ spanning arborescences rooted at $r$ if and only if it is rooted $k$-edge connected, i.e., all $r$-$v$ cuts have size at least $k$ for all $v\in V$.
Due to the typically uniform connectivity requirement, such results may well be used to decompose solutions of the Held-Karp relaxation (see \cite{genova_2017_experimental} for such an application).
In order to decompose a solution of the \ref{eq:rel_PCTSP}, where the variables $y_v$ indicate the non-uniform connectivities to the root $r$, we generalize ideas behind one of few results with non-uniform connectivity requirements due to \textcite{bang-jensen_1995_preserving}.
Concretely, we revise their splitting-off based construction to step away from decompositions into components that span the full subset of vertices with maximum connectivity (i.e., in our setting, those vertices with $y_v=1$), and instead only anchor each of the components at exactly two vertices in this subset, which we call the \emph{anchors} of the component.
This anchoring has the advantage that we can split each component of the decomposition into a path connecting its anchors (its \emph{backbone}) and trees incident to this path (its \emph{limbs}). By deleting or doubling the limbs we get even degree on all vertices except the anchors, which will prove useful to bound the cost of parity correction as mentioned earlier.

Besides being key to our improved approximation factor for \PCTSP{}, we remark that a common special case of our result and the one of \textcite{bang-jensen_1995_preserving} gives rise to simple algorithms with approximation guarantees matching those of the primal-dual approach of \textcite{goemans_1995_general}, both for \PCTSP{} and the related \emph{prize-collecting Steiner tree} problem, which is the following.

\begin{mdframed}[userdefinedwidth=0.95\linewidth]
{\textbf{Prize-Collecting Steiner Tree (\PCST{}\linkdest{prb:PCST}):}}
Given a graph $G=(V,E)$ with edge lengths $c_e\geq 0$ for all $e\in E$, a root $r \in V$, and penalties $\pi_v\geq 0$ for all $v\in V \setminus \{r\}$, the task is to find a tree $T=(V_T, E_T)$ in $G$ that contains the root $r$ and minimizes
$$\sum_{e\in E_T} c_{e} + \sum_{v\in V\setminus V_T} \pi_v\enspace.$$
\end{mdframed}
The classical LP formulation for \PCST{} is almost identical to the \ref{eq:rel_PCTSP}, except that the degree constraints are dropped, and the coefficients $2$ on the right hand side of the cut constraints are replaced by a~$1$.
For \PCTSP{} and \PCST{}, the algorithms of \citeauthor{goemans_1995_general} based on their primal-dual framework return cycles and trees, respectively, with length at most twice the $x$-cost of an optimum LP solution, and total penalty no more than once the $y$-cost of an optimum LP solution.
Matching the resulting guarantee on the sum of both, we present simple $2$-approximations for either of the problems.
We note that the currently best known approximation factor for \PCST{} is $1.968$ \cite{archer_2011_improved}, hence the proposed algorithm is inferior.
One particular reason why our new techniques for \PCTSP{} do not lead to an improved approximation guarantee for \PCST{} is that the LP formulation for \PCST{} has a rather weak integrality gap of $2$.
Still, we believe that our techniques enrich the study and may prove useful again also for \PCST{}, potentially in the context of provably stronger LP relaxations for \PCST{}.
One natural option is to adapt the bidirected cut relaxation for the Steiner tree problem to \PCST{}. However, its a big open question whether the integrality gap of the bidirected cut relaxation is strictly less than $2$ (see, e.g., \cite{rajagopalan_1999_bidirected, chakrabarty_2008_new}).

\subsection{Further related work}

The prize collection traveling salesperson problem discussed in this paper can be considered the main variant among several problems of similar nature that emerged from work of \citeauthor{balas_1989_prize} on a model for scheduling the daily operation of a steel rolling mill \cite{balas_1989_prize}.
In his original problem formulation, on top of the \PCTSP{} definition used here, there was a prize money $w_v \geq 0$ for every $v\in V \setminus \{r\}$ that the traveling salesperson collects upon visiting $v$, and a constraint that a prescribed amount $w_0$ of total prize money should be collected.
To distinguish from the \PCTSP{} formulation used here, we refer to Balas' formulation as the \emph{Quota PCTSP}, also in view of the special case of \emph{Quota TSP} that is obtained by setting $\pi_v=0$ for all $v\in V \setminus \{r\}$.
Quota TSP admits a $5$-approximation \cite{ausiello_2000_salesmen}, and by concatenating solutions of the Quota TSP and the \PCTSP{} subproblem, a constant factor approximation for Quota PCTSP is readily obtained \cite{ausiello_2007_prize}.
The special case of Quota TSP where $w_v = 1$ for each $v \in V \setminus \{r\}$ and $w_0=k$ for an integer $k$ is typically referred to as \emph{$k$-TSP} in the literature, and admits a $2$-approximation \cite{garg_2005_saving}.
Another budget-driven version of \PCTSP{} is \emph{budgeted prize-collecting TSP}.
In this problem, an upper bound on the allowed traveling distance is given, and the goal is to maximize the sum of all collected prizes.
\textcite{paul_2020_budgeted} provide a $2$-approximation algorithm for this problem, again assuming metric edge lengths as in all these variants.
Searching for a path instead of a cycle, and again imposing an upper bound on the traveling distance, results in the \emph{Orienteering} problem, which is as well widely studied in the literature in terms of approximation algorithms (e.g., \cite{blum_2003_approx,chekuri_2012_improved,dezfuli_2022_combinatorial}).

\PCTSP{} was also studied in special metric spaces.
For graph metrics in planar graphs, \textcite{bateni_2011_planar} show the existence of a PTAS. For Euclidean distances, a PTAS is known as well \cite{chan_2020_unified}.
For the asymmetric version of \PCTSP{}, \textcite{nguyen_2013_asymmetric} provides a $\lceil \log(n) \rceil$-approximation algorithm, where $n$ denotes the number of vertices.

Prize-collecting variations are also popular with other network design problems, in particular the Steiner tree problem mentioned above.
A further generalization of this problem is the \emph{prize-collecting Steiner forest problem} (also called prize-collecting generalized Steiner tree problem), in which we are given penalties for vertex pairs that need to be paid in case the vertices are not connected in the output graph.
By using a straightforward adaption of the threshold rounding approach from \cite{goemans_2009_combining}, one can get a $\sfrac{1}{ (1-e^{-\sfrac12})}\approx 2.541$-approximation \cite{hajiaghayi}.
This also upper bounds the integrality ratio of the natural LP relaxation for this problem.
\textcite{koenemann_2017_intgap} showed a lower bound of $\sfrac94$ on the integrality gap, which exceeds the integrality gap of $2$ for the corresponding LP relaxation of the \emph{Steiner forest problem}.

\subsection{Organization of the paper}

In \cref{sec:our_approach}, we give a detailed overview of a randomized version of our new algorithm, including in particular our decomposition lemma, the key technical tool that allows for our improved approximation guarantee.
\cref{sec:21} showcases the simple $2$-approximation algorithms for \PCTSP{} and \PCST{}.
In \cref{sec:splitting}, we give a proof of the decomposition lemma.
\cref{sec:randomization}, besides providing details on a randomized step of our algorithm, shows how to transform our approach to a deterministic one at no qualitative loss.
Finally, we conclude the paper with a deferred technical proof in \cref{sec:randomized_analysis}.

\section{Our approach}\label{sec:our_approach}

On a high level, we obtain our improved approximation guarantee for \PCTSP{} by following a common recipe for constructing tours:
\begin{enumerate*}
\item Construct a connected subgraph $H$ of the input graph (typically a spanning tree),
\item add a shortest $\odd(H)$-join $J$ to $H$,\footnote{We denote by $\odd(H)$ the set of all vertices of odd degree in $H$. Moreover, we remind the reader that for $Q\subseteq V$ with $|Q|$ even, a \emph{$Q$-join} is a set of edges with odd degrees precisely at vertices in $Q$.
}
and
\item return the cycle $C$ obtained by shortcutting an Eulerian tour in $H\cup J$.
\end{enumerate*}
In the final step, \emph{shortcutting} is to skip vertices that have been visited already when traversing the tour, i.e., resulting in an actual cycle.
To get bounds on the length of the thereby constructed cycle, it is enough to bound the length of $H$ and $J$.

A classical incarnation of the above recipe is the Christofides-Serdyukov algorithm for \TSP{} on a complete graph $G=(V,E)$ with metric edge lengths $c\colon E\to\mathbb{R}_{\geq 0}$.
\textcite{wolsey_1980_heuristic} proposed the following LP-based analysis:
For any point $x$ feasible for the \emph{Held-Karp relaxation}, whose feasible region is given by
\begin{equation*}
\PHK(G) \coloneqq \left\{x\in\mathbb{R}^E_{\geq 0}\colon
\begin{array}{rl}
	x(\delta(v)) =2\ & \forall v \in V \\
	x(\delta(S))\geq 2\ & \forall S\subsetneq V,\, S\neq\emptyset
\end{array}\right\}\enspace,
\end{equation*}
one can prove that a shortest spanning tree on $V$ has cost at most $c^\top x$, and a shortest $\odd(T)$-join has cost at most $\frac12c^\top x$, therefore proving a $\sfrac32$-approximation guarantee with respect to a shortest solution $x\in\PHK$.

\subsection{Review of the classical threshold rounding algorithm}\label{sec:classical_threshold_rounding}

Wolsey's analysis can be leveraged for an analysis of the threshold rounding algorithm by \textcite{bienstock_1993_note} as follows.
Starting from a solution $(x,y)$ of the \ref{eq:rel_PCTSP}, let $V_1\coloneqq \{v\in V\colon y_v=1\}$.
Through applying \emph{splitting off} techniques to $x$ (see \cref{sec:splitting} for details), we can obtain a point $z\in\PHK(G_1)$ of total length at most $c^\top x$, where $G_1=(V_1,E_1)$ is the complete graph on the vertex subset $V_1$.
Hence, Wolsey's analysis shows that we can get a cycle on $V_1$ of length at most $\frac32 c^\top z\leq \frac32 c^\top x$.
Because the loss in the penalty term incurred by this cycle may be arbitrarily large compared to the LP penalty $\sum_{v\in V}(1-y_v)\pi_v$, the idea of threshold rounding is to not directly apply the above reasoning to an optimal solution of the \ref{eq:rel_PCTSP}, but a scaled version thereof that can be obtained through the following lemma.

\begin{lemma}\label{lem:PCTSP_sol_properties}
	Let $(x, y)$ be a feasible solution of the \ref{eq:rel_PCTSP}.
	For every $\lambda\geq 0$, we can efficiently obtain a feasible solution $(x_\lambda, y_\lambda)$ of the \ref{eq:rel_PCTSP} such that $c^\top  x_{\lambda}\leq \lambda\cdot c^\top x$ and
	$$ y_{\lambda,v} = \min\{1, \lambda y_v\}\quad\text{for all } v\in V \enspace. $$
\end{lemma}

Simply speaking, \cref{lem:PCTSP_sol_properties} states that by investing a factor of $\lambda$ more on the $x$-part of the solution, we can also boost the connectivities $y$ by the same factor.
We remark that \cref{lem:PCTSP_sol_properties} was already proved in \cite{bienstock_1993_note} for optimal solutions of the \eqref{eq:PCST_relaxation} using the parsimonious property introduced by \textcite{parsi}; we repeat it here for general feasible solutions $(x,y)$, and provide a short proof in \cref{sec:splitting} for completeness.
To complete the argument for threshold rounding started above, let $(x^*, y^*)$ be an optimal solution of the \ref{eq:rel_PCTSP}, and for a threshold $\gamma\in(0,1]$, let $(x,y)$ be the solution obtained from $(x^*, y^*)$ through \cref{lem:PCTSP_sol_properties} with $\lambda=\sfrac1\gamma$.
We have seen that we can obtain a cycle $C=(V_C, E_C)$ of cost
\begin{equation}\label{eq:thresholding_tourCost}
c(E_C)\leq \frac32c^\top x \leq \frac3{2\gamma}\cdot c^\top x^*
\end{equation}
on $V_1=\{v\in V\colon y_v = 1\}=\{v\in V\colon y_v^*\geq \gamma\}$.
Additionally, we can now bound the penalty by
$
\pi(V\setminus V_C) = \sum_{v\notin V_1}\pi_v \leq \frac{1}{1-\gamma}\cdot \sum_{v\in V}(1-y_v^*)\pi_v
$.
Choosing $\gamma=\sfrac35$ balances the coefficients $\sfrac3{2\gamma}$ and $\sfrac1{(1-\gamma)}$, and gives the upper bound of $\sfrac52$ on the approximation factor, as showed in \cite{bienstock_1993_note}.

\subsection{A refined tour construction}\label{sec:decomposition}

As already mentioned in the introduction, our refined tour construction allows us to include some vertices in the output cycle that are not contained in our core set $V_1$ while preserving the upper bound on the length of the cycle.
To do so, we exploit a weak spot in the analysis of threshold rounding, namely the slack between the cost $c^\top z$ of the point $z\in\PHK(G_1)$ defined in \cref{sec:classical_threshold_rounding} and the cost $c^\top x$.
Transforming $x$ to $z\in\PHK(G_1)$ is done by splitting off operations.
In a single such operation, for some $\delta>0$, we decrease the value of $x$ on two well-chosen edges $\{v,u\}$ and $\{v,w\}$ by $\delta$, and increase the value on $\{u,w\}$ by $\delta$.
Thereby, the total cost decreases by $\delta c_{\{v,u\}} + \delta c_{\{v,w\}} - \delta c_{\{u,w\}}$, which is non-negative because $c$ satisfies the triangle inequality.
Viewed differently, we may exploit a budget of $c_{\{v,u\}} + c_{\{v,w\}}$ instead of $c_{\{u,w\}}$ for a $\delta$-fraction of the edge $\{u,w\}$ in $z$, i.e., if $\{u,w\}$ is chosen to be part of a tour, we can afford to replace it by the path $u$-$v$-$w$ with a certain probability, thereby gaining extra coverage of the vertex $v$ while maintaining the total cost bound $c^\top x$ in expectation.
This idea motivates backtracking the splitting off operations done to reach $z$, so that we can assign a certain budget to every component of $z$.
Doing so in a careful way results in our main technical lemma, which is \cref{lem:solution_decomposition} below (a formal proof is deferred to \cref{sec:splitting}).
Before actually stating the lemma, let us discuss an assumption on feasible solutions of the \ref{eq:rel_PCTSP} that we enforce repeatedly throughout the paper and which will prove useful later on.

\begin{assumption}\label{ass:e_0}
For the feasible solution $(x,y)$ of the \ref{eq:rel_PCTSP} and $V_1\coloneqq \{v\in V\colon y_v=1\}$, there is an edge $e_0 \in \binom{V_1}{2}$ with $x_{e_0}\geq1$.
\end{assumption}

We remark that \cref{ass:e_0} can be made without loss of generality.
Indeed, we can leverage a trick used previously, for example by \textcite{karlin_2021_slightly}, namely dividing the root node $r$ and all incident edges into two (with penalty zero and equal edge costs, respectively), and adding an edge $e_0$ of length $c(e_0)=0$ between the root and its copy.
Note that an LP solution for the original instance can be transformed to an LP solution of the same value for the thereby obtained auxiliary instance by setting $y_{v}=1$ for the copy of the root, $x_{e_0}=2-\frac{x(\delta(r))}{2}$, and distributing the $x$-weight on every edge incident to the root equally to its two copies in the auxiliary graph.
(Note that $x_{e_0} > 1$ if $x(\delta(r)) < 2$.)
Conversely, by shortcutting, we can transform a \PCTSP{} solution in the auxiliary instance to a \PCTSP{} solution of at most the same value in the original instance.

\begin{lemma}[Decomposition lemma]\label{lem:solution_decomposition}
	Let $(x, y)$ be a feasible solution of the \ref{eq:rel_PCTSP} satisfying \cref{ass:e_0} with edge $e_0$.
	We can in polynomial time construct a set $\mathcal{T}$ of trees and weights $\mu \in [0,1]^\mathcal{T}$ with the following properties:
	\begin{enumerate}

		\item\label{item:solution_partitioning} The solution $x$ is a conic combination of the trees in $\mathcal{T}$ with weights $\mu$ and the edge $e_0$, i.e.,
		$$
		x = \sum_{T\in \mathcal{T}}\mu_T\chi^{E[T]} + \chi^{e_0} \enspace.
		$$

		\item\label{item:incident_trees} For every $v\in V \setminus V_1$,
		$$
		\sum_{T\in \mathcal{T}\colon v\in V[T]} \mu_T = y_v \enspace.
		$$

		\item\label{item:anchor_points} For every $T\in\mathcal{T}$, we have $|V[T]\cap V_1|=2$, and we call the vertices in $V[T]\cap V_1$ the \emph{anchors} of $T$.

		\item\label{item:subtour_property} For $T \in \mathcal{T}$, let $e_T\coloneqq V[T]\cap V_1$ denote the edge joining the anchors of $T$, and let $G_1\coloneqq (V_1, E_1)$ be the multigraph with edge set $E_1\coloneqq \{e_0\} \cup \{e_T\colon T\in \mathcal{T}\}$.
		Then
		$$
		z \coloneqq \sum_{T\in \mathcal{T}}\mu_T\chi^{e_T} + \chi^{e_0} \in \PHK(G_1)\enspace.
		$$
	\end{enumerate}
\end{lemma}

Note that \cref{item:incident_trees} in \cref{lem:solution_decomposition} guarantees that the connectivity $y_v$ of vertices $v\in V\setminus V_1$ is reflected by the (weighted) number of trees incident to $v$ in our decomposition.
The latter will be crucial to obtain bounds on the penalties paid by our solutions for vertices in $V\setminus V_1$, and guaranteeing this property requires caution when constructing the trees in $\mathcal{T}$.
We also remark that in our proof, the point $z\in\PHK(G_1)$ in \cref{lem:solution_decomposition}\,\ref{item:subtour_property} will be a point that can be obtained from $x$ through splitting off operations.
\cref{fig:treeDecomposition} shows an example of our decomposition.

\begin{figure}[ht]
\hfill
\begin{subfigure}[t]{.4\linewidth}
\centering
\begin{tikzpicture}

\begin{scope}[scale = 3, shift={(0,-0)},rotate=90]
\node[Vertex] (v) at (0,0) {};
\node[Vertex] (u) at (0,2) {};

\node[above left=-4pt] at (0.05,0) {$v$};
\node[above right=-4pt] at (0.05,2) {$u$};

\node[Vertex2] (r) at (0,0.6) {};
\node[Vertex2] (l) at (0,1.4) {};
\node[Vertex2] (rb) at (-0.1,0.8) {};
\node[Vertex2] (rt) at (0.1,0.8) {};
\node[Vertex2] (lb) at (-0.1,1.2) {};
\node[Vertex2] (lt) at (0.1,1.2) {};

\draw[thick] (u)--(l);
\draw[thick] (v)--(r);
\draw[thick] (lt)--(rt);
\draw[thick] (lb)--(rb);

\draw[thick,dashed] (l)--(lb)--(lt)--(l);
\draw[thick,dashed] (r)--(rb)--(rt)--(r);

\node[Vertex2] (b1) at (-0.1,0.933) {};
\node[Vertex2] (b2) at (-0.1,1.067) {};
\node[Vertex2] (t1) at (0.1,0.933) {};
\node[Vertex2] (t2) at (0.1,1.067) {};

\begin{scope}[every path/.style={thick}]
\draw (u) --++ (140:0.1);
\draw (u) --++ (90:0.1);
\draw (u) --++ (40:0.1);

\draw (v) --++ (-40:0.1);
\draw (v) --++ (-90:0.1);
\draw (v) --++ (-140:0.1);
\end{scope}
\end{scope}

\end{tikzpicture}
 \subcaption{We assume $y_u=y_v=1$, weight $\sfrac12$ on solid edges, weight $\sfrac14$ on dashed edges, and weight $0$ on all other edges incident to shown vertices $w$ different from $u$ and $v$. This gives $y_w=\sfrac12$ for all such $w$. }
\end{subfigure}
\hfill
\begin{subfigure}[t]{.4\linewidth}
\centering
\begin{tikzpicture}

\tikzstyle{Vertex}=[circle,fill=black,minimum size=6,inner sep=0pt]
\tikzstyle{Vertex2}=[draw=black,circle,fill=white,minimum size=4,inner sep=0pt]

\begin{scope}[scale = 3,shift={(3,0)},rotate=90]
\node[Vertex] (v) at (0,0) {};
\node[Vertex] (u) at (0,2) {};

\node[above left=-4pt] at (0.05,0) {$v$};
\node[above right=-4pt] at (0.05,2) {$u$};

\node[Vertex2] (r) at (0,0.6) {};
\node[Vertex2] (l) at (0,1.4) {};
\node[Vertex2] (rb) at (-0.1,0.8) {};
\node[Vertex2] (rt) at (0.1,0.8) {};
\node[Vertex2] (lb) at (-0.1,1.2) {};
\node[Vertex2] (lt) at (0.1,1.2) {};

\node[Vertex2] (b1) at (-0.1,0.933) {};
\node[Vertex2] (b2) at (-0.1,1.067) {};
\node[Vertex2] (t1) at (0.1,0.933) {};
\node[Vertex2] (t2) at (0.1,1.067) {};

\draw[red, thick] (v) to[bend left = 10] (r)--(r)--(rb) to[bend left = 20] (b1) --(b1) to[bend left = 20] (b2)--(b2) to [bend left = 20] (lb)--(lb) -- (l) to[bend left = 10] (u);
\draw[red,thick] (rb)--(rt) to[bend left = 20] (t1)--(t1) to[bend left = 20] (t2)--(t2) to[bend left = 20] (lt);

\draw[blue, thick] (v) to[bend right = 10] (r)--(r)--(rt) to[bend right=20] (t1) --(t1) to[bend right=20] (t2)--(t2) to [bend right=20] (lt)--(lt) -- (l) to[bend right=10] (u);
\draw[blue,thick] (lt)--(lb) to[bend left=20] (b2)--(b2) to[bend left=20] (b1)--(b1) to[bend left=20] (rb);

\begin{scope}[every path/.style={thick}]
\draw (u) --++ (140:0.1);
\draw (u) --++ (90:0.1);
\draw (u) --++ (40:0.1);

\draw (v) --++ (-40:0.1);
\draw (v) --++ (-90:0.1);
\draw (v) --++ (-140:0.1);
\end{scope}

\end{scope}

\end{tikzpicture} \subcaption{\cref{lem:solution_decomposition} decomposes the shown part of the LP solution into the red and the blue tree anchored at $u$ and $v$, each of weight $\sfrac14$. Note that it is not possible to decompose this part into $u$-$v$ walks satisfying \cref{lem:solution_decomposition}\,\ref{item:incident_trees}.}
\end{subfigure}
\hfill~
\caption{An excerpt of a solution of the \ref{eq:rel_PCTSP} and its decomposition given by \cref{lem:solution_decomposition}.}
\label{fig:treeDecomposition}
\end{figure}
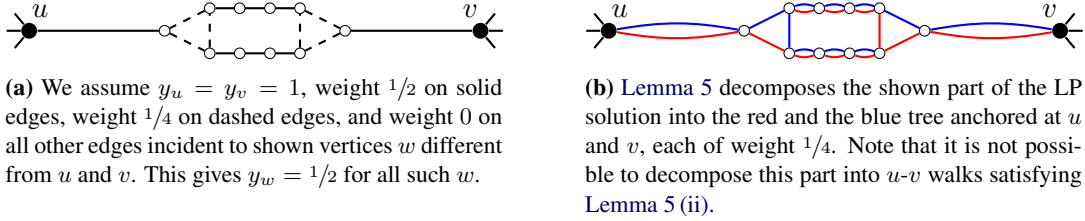

Our decomposition gives a more fine-grained analysis of the bound on the edge cost in threshold rounding than we obtained in \eqref{eq:thresholding_tourCost}:
We have
\begin{equation}\label{eq:reiterate_Wolsey}
c(E_C)
\leq \frac32 c^\top z
= \frac32 \sum_{T\in\mathcal{T}} \mu_T c_{e_T}
\leq \frac32 \sum_{T \in \mathcal{T}}\mu_T c(E[T])
\leq \frac32 c^\top x\enspace.
\end{equation}
Here, in the middle inequality, we are bounding $c_{e_T}$ by $c(E[T])$.
If $T$ was a path connecting the endpoints of $e_T$, we could exploit this by, whenever $e_T$ appears in a tour that we construct, replacing $e_T$ with $T$, and thereby avoid penalty at all vertices on the path without affecting the upper bound in \eqref{eq:reiterate_Wolsey}.
In the general case, in order to maintain a tour, we can only plug in the path in $T$ connecting the endpoints of $e_T$.
This path and its complement are important notions that we use repeatedly, hence we define the following (also see \cref{fig:backbones_and_limbs} for an illustration).

\begin{definition}
Let $T$ be a tree, and let $V_1$ be a vertex set such that $|V[T]\cap V_1|=2$.
The \emph{anchors} of $T$ in $V_1$ are the two vertices in $V[T]\cap V_1$.
The \emph{backbone} of $T$ with respect to $V_1$, denoted by $\bb{T}$, is the set of edges of $T$ on the unique path connecting the two vertices in $V[T]\cap V_1$.
The \emph{limbs} of $T$ with respect to $V_1$, denoted by $\lb{T}$, is the set of edges of $T$ that are not in the backbone of $T$ with respect to $V_1$.
\end{definition}

\begin{figure}[!ht]
\hfill
\begin{subfigure}[t]{0.26\linewidth}
\begin{tikzpicture}[xscale=0.9]

\begin{scope}[scale=0.4,shift={(0,0)}]
\node[Vertex] (u) at (0,0) {};
\node[Vertex] (v) at (10,1) {};

\node[above=1pt] at (v) {$v$};
\node[above=1pt] at (u) {$u$};

\node[Vertex2] (b1) at (3,1) {};
\node[Vertex2] (b2) at (6,0.5) {};
\node[Vertex2] (b3) at (8,0) {};
\node[Vertex2] (l1) at (3.5,2) {};
\node[Vertex2] (l2) at (4.5,2.5) {};
\node[Vertex2] (l3) at (3,3) {};
\node[Vertex2] (l4) at (2,2.5) {};
\node[Vertex2] (l5) at (7.5,-1) {};
\node[Vertex2] (l6) at (8.5,-1.5) {};

\draw[red, thick] (u)--(b1)--(b2)--(b3)--(v);
\draw[blue, thick] (b1)--(l1)--(l2);
\draw[blue, thick] (l1)--(l3)--(l4);
\draw[blue, thick] (b3)--(l5)--(l6);

\end{scope}

\end{tikzpicture} \subcaption{A tree $T$ with anchors $u$ and $v$, red backbone, and blue limbs.}
\end{subfigure}
\hfill
\begin{subfigure}[t]{0.54\linewidth}
\centering
\begin{tikzpicture}[xscale=0.9]

\begin{scope}[scale=0.4]
\node[Vertex] (u) at (0,0) {};
\node[Vertex] (v) at (10,1) {};

\node[above=1pt] at (v) {$v$};
\node[above=1pt] at (u) {$u$};

\node[Vertex2] (b1) at (3,1) {};
\node[Vertex2] (b2) at (6,0.5) {};
\node[Vertex2] (b3) at (8,0) {};
\node[Vertex2] (l1) at (3.5,2) {};
\node[Vertex2] (l2) at (4.5,2.5) {};
\node[Vertex2] (l3) at (3,3) {};
\node[Vertex2] (l4) at (2,2.5) {};
\node[Vertex2] (l5) at (7.5,-1) {};
\node[Vertex2] (l6) at (8.5,-1.5) {};

\draw[red, thick] (u)--(b1)--(b2)--(b3)--(v);
\end{scope}

\begin{scope}[scale=0.4,xshift=13cm]
\node[Vertex] (u) at (0,0) {};
\node[Vertex] (v) at (10,1) {};

\node[above=1pt] at (v) {$v$};
\node[above=1pt] at (u) {$u$};

\node[Vertex2] (b1) at (3,1) {};
\node[Vertex2] (b2) at (6,0.5) {};
\node[Vertex2] (b3) at (8,0) {};
\node[Vertex2] (l1) at (3.5,2) {};
\node[Vertex2] (l2) at (4.5,2.5) {};
\node[Vertex2] (l3) at (3,3) {};
\node[Vertex2] (l4) at (2,2.5) {};
\node[Vertex2] (l5) at (7.5,-1) {};
\node[Vertex2] (l6) at (8.5,-1.5) {};

\draw[red, thick] (u)--(b1)--(b2)--(b3)--(v);
\draw[blue, thick] (b1) to[bend left=20] (l1)--(l1) to [bend left=20] (l2)
--(l2) to[bend left = 20] (l1)--(l1) to[bend left=20] (b1);
\draw[blue, thick] (l1) to[bend left=20] (l3)--(l3) to [bend left=20] (l4)
--(l4) to[bend left = 20] (l3)--(l3) to[bend left=20] (l1);
\draw[blue, thick] (b3) to[bend left=20] (l5)--(l5) to [bend left=20] (l6)
--(l6) to[bend left = 20] (l5)--(l5) to[bend left=20] (b3);
\end{scope}

\end{tikzpicture}
\subcaption{The two walks generated from $T$: One consisting of the backbone only (left), and one with two copies of the limbs added on top (right).}
\end{subfigure}
\hfill~
\caption{Anchors, backbone, and limb edges of a tree, and the two walks that can be constructed thereof.}
\label{fig:backbones_and_limbs}
\end{figure}
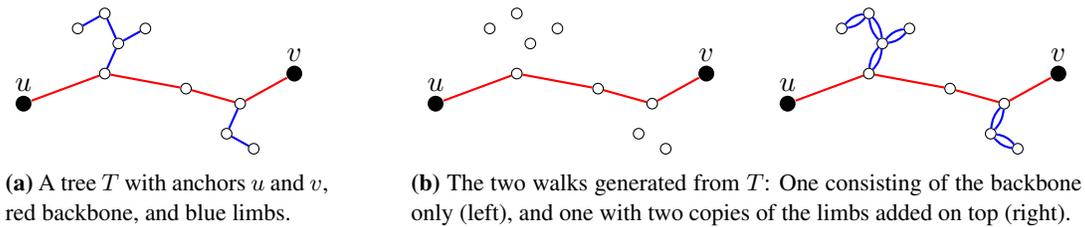

Using this new notation, we can strengthen the bound in \eqref{eq:reiterate_Wolsey} to
\begin{equation}\label{eq:stronger_tourBound}
c(E_C) \leq
\frac32 c^\top z
\leq \frac32 \sum_{T\in \mathcal{T}} \mu_Tc(\bb{T})\enspace.
\end{equation}
At the above cost, we can build a tour on the vertices in $V_1$ and replace edges of that tour with the corresponding backbone paths.
This extends the set of vertices that are covered beyond $V_1$ by those lying on the added backbones.
Moreover, compared to the upper bound $\frac32 c^\top x$, we still have a budget of $\frac32 \sum_{T \in \mathcal{T}} \mu_Tc(\lb{T})$ remaining.
We can use this budget to add two copies of the limb edges of some trees to the cycle $C$ in order to cover even more vertices.
Note that by adding two copies, we make sure that the resulting edge set can again be shortcut to a cycle.
Both of these observations contribute to saving penalties.

Illustrating how to translate the above intuition to actually improve coverage of vertices not in $V_1$ is easiest through exploiting randomization.
Thus, we present a probabilistic approach here, and only later show that derandomization is possible.
To build a basis for our tour, we exploit \cref{ass:e_0} and the structure obtained thereby:
\cref{lem:solution_decomposition}\,\ref{item:subtour_property} implies that $z - \chi^{e_0}=\sum_{T\in\mathcal{T}}\mu_T\chi^T$ is in the \emph{spanning tree polytope} $\PST(G_1)$ over $G_1$.\footnote{Recall that the spanning tree polytope $\PST(G)$ over a graph $G=(V,E)$ is the convex hull of all incidence vectors of spanning trees in $G$, and can be described by $\PST(G)=\{x\in\mathbb{R}^E_{\geq 0}\colon x(E) = |V|-1,\, x(E[S])\leq |S|-1\, \forall S\subseteq V\}$.}
Being a matroid base polytope, we can exploit known marginal-preserving negatively correlated randomized rounding schemes in $\PST(G_1)$, concretely \emph{randomized pipage rounding} (see \cref{sec:randomization} for details).
Based on this, our approach is the following:
\begin{stepsroman}
\item\label{step:sampling} Sample a spanning tree of $G_1$ with marginals $z-\chi^{e_0}$, such that edges $e_T$ appear with probability $\mu_T$.
\item\label{step:pathReplacement} Replace every sampled edge $e_T$ by the path with edge set $\bb{T}$, and with probability $\sfrac34$ (independently for all sampled $e_T$), add two copies of $\lb{T}$.
\item\label{step:parityCorrection} Correct parities in the resulting graph using an appropriate join, and shortcut.
\end{stepsroman}

We remark that---instead of randomized pipage rounding---one could sample the spanning tree in \cref{step:sampling} from a maximum entropy distribution and use the framework introduced by \citeauthor{karlin_2021_slightly} in their seminal $(1.5-\varepsilon)$-approximation algorithm for \TSP{} \cite{karlin_2021_slightly, karlin_derand} to slightly reduce the cost of parity correction in \cref{step:parityCorrection}.
However, the order of magnitude of this improvement (the authors show a bound of $\varepsilon > 10^{-36}$) is negligible in comparison with the loss that we incur through our analysis (see \cref{remark:threshold_choice}).
Additionally, derandomization of pipage rounding is rather straightforward in our setting using known tools (see \cref{sec:randomization}).

Replacing each sampled tree edge $e_T$ by the corresponding backbone together with parity correction leads to precisely the bound on the total edge cost in \eqref{eq:stronger_tourBound}, and adding limb edges probabilistically as in \cref{step:pathReplacement} uses up the remaining budget $\frac32 \sum_{T \in \mathcal{T}} \mu_Tc(\lb{T})$ in expectation.
Adding two copies of $\lb{T}$ implies that the edge set we replace $e_T$ by can always be cast as a walk connecting the two endpoints of $e_T$, hence vertices in $V\setminus V_1$ will always have even degree.
Consequently, in \cref{step:parityCorrection}, parities only need to be corrected at vertices in $V_1$, and hence at a cost that we can still bound by $\frac12c^\top z$.
Moreover, we can make use of the negative correlation property of the sampling procedure to prove bounds on the probability that a vertex $v\in V\setminus V_1$ is not covered by any of the used walks, and thereby obtain an upper bound on the expected penalty that we incur.
(Note that in \cref{step:parityCorrection}, we may use a join that could possibly connect further vertices, but we do not exploit this in our analysis.)
Summarizing \cref{step:sampling,step:pathReplacement} leads to the following lemma, whose formal proof we defer to \cref{sec:randomization}.

\begin{lemma}\label{lem:walk_sampling}
Let $(x, y)$ be a feasible solution of the \ref{eq:rel_PCTSP} satisfying \cref{ass:e_0}.
Let $\mathcal{T}$ and $\mu\in [0,1]^\mathcal{T}$ be a set of trees and weights with the properties guaranteed by \cref{lem:solution_decomposition}.
We can efficiently sample a multigraph $H$ on a subset of $V$ with the following properties:
\begin{enumerate}
\item $H$ is connected, spans $V_1\coloneqq \{v\in V\colon y_v=1\}$, and all vertices in $V\setminus V_1$ have even degree.

\item\label{item:prob_H_bound} We have
\begin{multline*}
\Exp\left[c(E[H]) + \pi(V\setminus V[H])\right] \\
\leq \sum_{T \in \mathcal{T}} \mu_T \cdot \left( c(\bb{T}) + \frac32 \cdot c(\lb{T}) \right) + \sum_{v\in V\setminus V_1}\pi_v \exp\left(-\frac{3y_v}{4}\right)\enspace.
\end{multline*}
\end{enumerate}
\end{lemma}

With this result at hand, we can state a randomized version of our new algorithm, \cref{alg:treeSampling}.
Note that the input only requires a feasible solution $(x^*, y^*)$ of the \ref{eq:rel_PCTSP}, even though typically, we will call \cref{alg:treeSampling} on an optimal solution. Moreover, the choice of a suitable threshold $\gamma\in(0,1]$ is left open for the moment; we show how to optimally exploit this remaining flexibility later on.

\begin{algorithm2e}[!thb]
\caption{Our new randomized algorithm for \PCTSP.\strut{}}\label{alg:treeSampling}
\SetKwInOut{Input}{Input}
\Input{Feasible solution $(x^*,y^*)$ of the \ref{eq:rel_PCTSP}, threshold $\gamma\in (0, 1]$.}
\smallskip
\begin{stepsarabic}

\item Guarantee that \cref{ass:e_0} is satisfied by modifying the instance accordingly if necessary.

\item Apply \cref{lem:PCTSP_sol_properties} to $(x^*, y^*)$ with $\lambda=\sfrac{1}{\gamma}$ to obtain another feasible solution $(x, y)$.

\item Compute a set of trees $\mathcal{T}$ with weights $\mu \in [0,1]^\mathcal{T}$ by applying \cref{lem:solution_decomposition} to $(x,y)$.

\item\label{algitem:sample_walks} Sample a multigraph $H$ on $V$ through \cref{lem:walk_sampling} applied to $(x, y)$ and $(\mathcal{T}, \mu)$.

\item\label{algitem:parity_correction}  Compute a shortest $\odd(H)$-join $J$.

\item Let $C$ be a cycle obtained by shortcutting an Eulerian tour in $H\cup J$.

\end{stepsarabic}
\Return{$C$.}
\end{algorithm2e}

Using \cref{lem:walk_sampling}, we get the following guarantees.
We remark that compared to classical threshold rounding, the factor $\sfrac3{2\gamma}$ on the term $c^\top x^*$ is unchanged, while at the same time, we get clearly improved factors on the penalty side.

\begin{theorem}\label{thm:treeSamplingAnalysis}
Let $(x^*, y^*)$ be a feasible solution of the \ref{eq:rel_PCTSP}.
On input $(x^*,y^*)$ and $\gamma \in (0,1]$,
\cref{alg:treeSampling} returns, in polynomial time, a cycle $C=(V_C, E_C)$ such that
\begin{equation*}
\Exp[c(E_C) + \pi(V\setminus V_C)] \le \frac{3}{2\gamma} \cdot c^\top x^* +  \sum_{v \in V\colon y^*_v < \gamma} \pi_v\cdot\exp\left(-\frac{3y_v^*}{4\gamma}\right)\enspace.
\end{equation*}
\end{theorem}

\begin{proof}
By \cref{lem:walk_sampling}, the multigraph $H$ constructed in \cref{alg:treeSampling} has odd degrees only at vertices in $V_1$, hence we can also correct parities by an $\odd(H)$-join in the multigraph $G_1=(V_1,E_1)$ with edge set $E_1= \{e_0\} \cup \{e_T\colon T\in \mathcal{T}\}$.
By \cref{item:subtour_property} of \cref{lem:solution_decomposition}, the point $z=\sum_{T\in\mathcal{T}}\mu_T \chi^{e_T} + \chi^{e_0}$ is in $\PHK(G_1)$.
Thus, following Wolsey's analysis, a shortest $\odd(H)$-join $J$ in $G_1$ has cost at most
$$
c(E[J])
\leq \frac12c^\top z
\leq  \frac12 \sum_{T\in\mathcal{T}}\mu_T c(\bb{T})\enspace.
$$
Note that $J$ also is a shortest $\odd(H)$-join in the original complete graph $G$.
Combining this with the bound of \cref{item:prob_H_bound} in \cref{lem:walk_sampling} immediately gives the claimed guarantee
\begin{align*}
\Exp[c(E_C) + \pi(V\setminus V_C)] &\leq \Exp[c(E[H]) + c(E[J]) + \pi(V\setminus V_C)]\\
&\leq \frac32c^\top x + \sum_{v\in V\setminus V_1}\pi_v \exp\left(-\frac{3y_v}{4}\right)
\leq \frac{3}{2\gamma} \cdot c^\top x^* +  \sum_{v \in V\colon y^*_v < \gamma} \pi_v\exp\left(-\frac{3y_v^*}{4\gamma}\right)
\enspace.
\end{align*}
Here, in the last inequality, we use that by \cref{lem:PCTSP_sol_properties}, we have $c^\top x \le \frac1\gamma c^\top x^*$, and also $y_v = \frac1\gamma y^*_v$ for $v\in V\setminus V_1$ because $V\setminus V_1 = \{v\in V\colon y^*_v < \gamma\}$.
To finish the proof, we note that indeed, all steps in \cref{alg:treeSampling} can be implemented in polynomial time.
\end{proof}

\subsection{Choosing the right threshold}\label{sec:choose_gamma}

Directly balancing the penalty terms in \cref{thm:treeSamplingAnalysis} by a proper choice of the threshold $\gamma$ gives the subsequent corollary.
We remark that the resulting approximation guarantee does not yet beat the previously best known one of $1.915$ by \textcite{goemans_2009_combining}, but it is significantly better than the standard analysis of threshold rounding, which gives a $\frac52$-approximation.

\begin{corollary}
Given a feasible solution $(x^*,y^*)$ of the \ref{eq:rel_PCTSP} and $\gamma = \sfrac{1}{\left(1+ \frac23\exp(-\sfrac34)\right)} \approx 0.761$, \cref{alg:treeSampling} returns in polynomial time a \PCTSP{} solution of expected value at most $\alpha \cdot (c^\top x^*+\pi^\top (1-y^*))$ with $\alpha = \frac32 + \exp(-\sfrac34) < 1.973$.
\end{corollary}

\begin{proof} For general $\gamma\in(0,1]$, we can bound the penalty term of the bound given in \cref{thm:treeSamplingAnalysis} as follows:
\begin{align*}
\sum_{v \in V\colon y^*_v < \gamma} \pi_v  \exp\left(-\frac{3y_v^*}{4\gamma}\right)
\le \max\left\{1,\frac{\exp\left({-\sfrac34}\right)}{1-\gamma}\right\} \cdot \sum_{v \in V\colon y^*_v < \gamma} \pi_v (1-y^*_v)\enspace.
\end{align*}
Here we used that for $y\in[0,\gamma]$, we have $\exp\big({-\frac{3y}{4\gamma}}\big) \leq \max\big\{1,\frac{\exp\left({-\sfrac34}\right)}{1-\gamma}\big\}\cdot(1-y)$.
The latter can be derived by noting that the inequality holds for $y\in \{0, \gamma\}$, and that $y\mapsto \exp\big({-\frac{3y}{4\gamma}}\big)$ is convex in $y$.
Thus, \cref{thm:treeSamplingAnalysis} implies
$$
\Exp[c(E_C)] + \Exp[\pi(V\setminus V_C)] \leq \max \left\{\frac{3}{2\gamma},1,\frac{\exp\left({-\sfrac34}\right)}{1-\gamma}\right\}\cdot\left(\sum_{e\in E} c_ex_e^* + \sum_{v\in V} \pi_v(1-y_v^*)\right)\enspace,
$$
where the maximum expression is minimized for $\gamma = \sfrac{1}{\left(1+ \frac23\exp(-\sfrac34)\right)}$, giving the claimed approximation factor with respect to the objective value of the input solution $(x^*, y^*)$ of the \ref{eq:rel_PCTSP}.
\end{proof}

When choosing the threshold $\gamma$ with respect to a specific distribution (similarly to \textcite{goemans_2009_combining}), we get a significant improvement of the best known approximation ratio. We remark that in contrast to the classical threshold rounding algorithm, it is no longer optimal to choose $\gamma$ with respect to a uniform distribution over some interval in $[0,1]$.

\begin{theorem}\label{thm:randomizedAnalysis}
	Let $b= 0.6945$. Sampling $\gamma$ from the interval $[b,1]$ using a distribution with density $f(\gamma)\propto\exp(-\sfrac b\gamma)$, and starting \cref{alg:treeSampling} from a feasible solution $(x^*, y^*)$ of the \ref{eq:rel_PCTSP}, we get in polynomial time a \PCTSP{} solution of expected value at most $\alpha\cdot\left(c^\top x^*+ \pi^\top(1-y^*)\right)$, where $\alpha<1.774$.
\end{theorem}

The proof of \cref{thm:randomizedAnalysis} is postponed to \cref{sec:randomized_analysis}.
We remark that in contrast to the classical threshold rounding approach, we do not profit from combining our algorithm with the primal-dual algorithm by \textcite{goemans_1995_general}. This makes our algorithm self-contained and easy to state.

Note that the \ref{eq:rel_PCTSP} can be solved in polynomial time through the ellipsoid method, as efficient separation over the constraints $x(\delta(S))\geq y_v$ can be reduced to minimum $r$-$v$ cut calculations for all $v\in V$ in the graph with capacities given by $x$.
Thus, applying \cref{thm:randomizedAnalysis} to an optimal solution $(x^*, y^*)$ of the \ref{eq:rel_PCTSP} directly implies \cref{thm:improvedApproximation} if one allows randomization.
In order to obtain a deterministic procedure, we show in \cref{sec:randomization} how to derandomize the proposed algorithm. Note that based on \cref{thm:randomizedAnalysis}, there are two steps where we exploit randomization.
First, we choose the threshold $\gamma$ from some distribution.
We will see that given an \emph{optimal} LP-solution $(x^*,y^*)$ of the \ref{eq:rel_PCTSP}, it suffices to try all thresholds in the set $\{y_v^*\colon v \in V\}$ to get the claimed bound (cf. \cref{sec:derandomizeGamma}).
Second, the construction of the multigraph $H$ in \cref{algitem:sample_walks} of \cref{alg:treeSampling} is based on randomized pipage rounding in the spanning tree polytope of an auxiliary graph.
We show in \cref{sec:derandomizeWalkSampling} that we can also use a deterministic version of pipage rounding to achieve the same guarantees.
This will prove \cref{thm:improvedApproximation}.

We make the reader aware that the first step above was immediate in the setting of \textcite{goemans_2009_combining}, because the values $\gamma=y_v^*$ for $v\in V$ already result in all possible supports of the tours that can be obtained through a threshold, and tours were subsequently built in a black-box way.
In our refined approach, this is no longer the case; other values of $\gamma$ may result in different tours, and it requires some thought to recover the conclusion that the best solution is obtained for some $\gamma=y_v^*$.

\begin{remark}\label{remark:threshold_choice}
Computational experiments (through discretizing the distribution to sample $\gamma$ from) indicate that even by choosing an optimal distribution for $\gamma$, one cannot push the above analysis to prove a bound on the approximation factor of value less than or equal to $1.773$. Consequently, the distribution over $\gamma$ that we propose in \cref{thm:randomizedAnalysis} is very close to optimal for an analysis with respect to the bounds given by \cref{thm:treeSamplingAnalysis}.
\end{remark}

\section[Simple \texorpdfstring{$2$}{2}-approximations for PCTSP and PCST through decompositions]{\boldmath Simple $2$-approximations for \PCTSP{} and \PCST{} through decompositions}\label{sec:21}

In this section, we show that for \PCTSP{} and \PCST{}, we can obtain an approximation guarantee similar to the one of the primal-dual approach by \textcite{goemans_1995_general} through very simple algorithms.
More precisely, the proposed algorithms return cycles and trees, respectively, of cost at most twice the $x$-cost of an optimum LP solution plus once the $y$-cost of an optimum LP solution.
The main ingredient giving rise to these algorithms is the following simplified variation of our decomposition lemma, which can also be obtained by utilizing an existential result on packing branchings in a directed multigraph by \textcite[Theorem 2.6]{bang-jensen_1995_preserving}.

\begin{lemma}\label{lem:fractional_tree_partition}
	Let $(x, y)$ be a feasible solution of the \ref{eq:rel_PCTSP}.
	We can in polynomial time compute a set of trees $\mathcal{T}$ containing the root $r$ and weights $\mu \in [0,1]^\mathcal{T}$ such that $\sum_{T\in\mathcal{T}}\mu_T = 1$,
    \begin{equation*}
		\sum_{T\in \mathcal{T}}\mu_T\chi^{E[T]} \leq x\enspace,
        \qquad\text{and}\qquad
		\forall v\in V\colon\
		\sum_{T\in\mathcal{T}\colon v\in V[T]}\mu_T = y_v\enspace.
    \end{equation*}
\end{lemma}

We prove \cref{lem:fractional_tree_partition} through a common generalization with \cref{lem:solution_decomposition} in \cref{sec:splitting}.
Based on \cref{lem:solution_decomposition}, the simple algorithm we propose for \PCTSP{} is the following.

\begin{algorithm2e}[!ht]
\caption{A simple $2$-approximation for \PCTSP\strut{}}\label{alg:treePartitionTours}
\SetKwInOut{Input}{Input}
\Input{\PCTSP{} instance $(G,r,c,\pi)$ on $G=(V,E)$.}
\smallskip
\begin{stepsarabic}
\item\label{step:sol} Compute an optimal solution $(x^*, y^*)$ of the \ref{eq:rel_PCTSP}.

\item Compute a set of trees $\mathcal{T}$ through \cref{lem:fractional_tree_partition} applied to $(x^*,y^*)$.

\item\label{algitem:doubling_tree} For each tree $T\in\mathcal{T}$, obtain a cycle $C_T$ through duplicating $T$ and shortcutting.

\item Let $S=\arg\min_{T\in\mathcal{T}}c(E[C_T])+\pi(V\setminus V[C_T])$.

\end{stepsarabic}
\Return{$C_S$.}
\end{algorithm2e}

\begin{theorem}\label{thm:21PCTSP}
\cref{alg:treePartitionTours} returns, in polynomial time, a cycle $C$ satisfying
$$ c(E[C]) + \pi(V \setminus V[C]) \leq 2\cdot c^\top x^* + \pi^\top (1-y^*)\enspace,$$
where $(x^*,y^*)$ is the optimal solution of the \ref{eq:rel_PCTSP} computed in \cref{step:sol} of the algorithm.
\end{theorem}

\begin{proof}
We do a probabilistic analysis of the algorithm.
The trees in $\mathcal{T}$ from \cref{lem:fractional_tree_partition} come with weights $\mu\in [0,1]^\mathcal{T}$ such that $\sum_{T \in \mathcal{T}} \mu_{T}=1$ and $\sum_{T\in\mathcal{T}} \mu_T \chi^{E[T]} \leq x^*$.
Sampling a tree $T$ from $\mathcal{T}$ with marginals $\mu$ gives
$$
\Exp[c(E[C_T])] =
\sum_{T \in \mathcal{T}} \mu_T \cdot 2c(E[T]) =
2\cdot c^\top\!\left(\sum_{T\in \mathcal{T}}\mu_T\chi^{E[T]}\right) \leq
2\cdot c^\top x^* \enspace.
$$
Similarly, for the penalty term we get
\begin{align*}
\Exp[\pi(V\setminus V[C_T])] &=
\sum_{v\in V} \Prob[v\notin T]\cdot \pi_v \\
&=
\sum_{v\in V} \left(1-\sum_{T\in\mathcal{T}\colon v\in V[T]} \mu_T\right) \cdot \pi_v =
\sum_{v\in V} (1-y_v^*)  \cdot \pi_v =
\pi^\top (1-y^*)\enspace.
\end{align*}
Hence, the expected objective value of $C_T$ can be bounded by $2\cdot c^\top x^* + \pi^\top (1-y^*)$.
Thus, also at least one of the cycles $C_T$ for $T\in\mathcal{T}$ satisfies the desired guarantee.
\end{proof}

For \PCST{}, we consider the following very related LP relaxation.
Recall that in \PCST{}, we do not assume the triangle inequality for distances, and thus also the LP relaxation comes without degree constraints.
\begin{equation}\tag{PCST LP relaxation}\label{eq:PCST_relaxation}
	\begin{array}{rrcll}
		\min & \displaystyle \sum_{e\in E}c_ex_e + \sum_{v\in V}\pi_v (1-y_v) \\
		& x(\delta(S)) & \geq & y_v & \forall S\subseteq V\setminus\{r\}, v\in S\\
		& y_r & = & 1 \\
		& x_e & \geq & 0 & \forall e\in E\\
		& y_v & \in & [0,1] & \forall v\in V\enspace.
	\end{array}
\end{equation}
The main step towards mimicking \cref{alg:treePartitionTours} is to transform a solution of the \ref{eq:PCST_relaxation} to one of the \ref{eq:rel_PCTSP} to be able to apply \cref{lem:fractional_tree_partition} in a black-box way, and later transform the obtained components back (see \cref{algitem:backtransform} of \cref{alg:PCST}). The latter is required because edge costs in \PCST{} are not necessarily metric.

\begin{algorithm2e}[!ht]
\caption{A simple $2$-approximation for \PCST\strut{}}\label{alg:PCST}
\SetKwInOut{Input}{Input}
\Input{\PCST{} instance $(G,r,c,\pi)$ on $G=(V,E)$.}
\smallskip
\begin{stepsarabic}
\item\label{step:sol_PCST} Compute an optimal solution $(x^*, y^*)$ of the \ref{eq:PCST_relaxation}.

\item Consider the \PCTSP{} instance $(\bar G, r, \bar c, \pi)$ where $\bar c$ denotes the metric closure of the edge lengths $c$ on the complete graph $\bar G$ on $V$. Obtain a solution $(x,y^*)$ of the \ref{eq:rel_PCTSP} from $(2x^*, y^*)$ by splitting off until all degree constraints are satisfied.

\item Compute a set of trees $\mathcal{T}$ through \cref{lem:fractional_tree_partition} applied to $(x,y^*)$.

\item\label{algitem:backtransform} For each $T\in\mathcal{T}$, let $T_0$ be obtained from $T$ through replacing edges by shortest paths connecting their endpoints in $G$.

\item Let $S=\arg\min_{T\in\mathcal{T}}c(E[T_0])+\pi(V\setminus V[T_0])$.

\end{stepsarabic}
\Return{a minimum spanning tree on $V[S_0]$.}
\end{algorithm2e}

\begin{theorem}\label{thm:21PCST}
\cref{alg:PCST} returns, in polynomial time, a tree $T$ satisfying
$$ c(E[T]) + \pi(V \setminus V[T]) \leq 2\cdot c^\top x^* + \pi^\top (1-y^*)\enspace,$$
where $(x^*,y^*)$ is the optimal solution of the \ref{eq:PCST_relaxation} computed in \cref{step:sol_PCST} of the algorithm.
\end{theorem}

\begin{proof}
As in \cref{alg:PCST}, let $\bar c$ denote the metric closure of the (not necessarily metric) edge lengths $c$.
Note that $\bar c^\top x \leq 2c^\top x^*$, as splitting off does not increase the total length with respect to metric lengths.
Additionally, for the connected multigraphs $T_0$ computed in \cref{algitem:backtransform}, we have $\bar c(E[T])=c(E[T_0])$ by definition.
With this at hand, we can again do a randomized analysis, where we sample a tree $T$ from $\mathcal{T}$ with marginals equal to the weights $\mu$ obtained through \cref{lem:fractional_tree_partition}.
This gives an expected objective value
\begin{align*}
\Exp[c(E[T_0]) + \pi(V\setminus V[T_0])]
&=
\Exp[\bar c(E[T])] + {\sum_{v\in V}} \Pr[v\notin T]\cdot \pi_v \\
&\leq
\sum_{T\in\mathcal{T}} \mu_T \bar c(E[T]) + \sum_{v\in V} \Bigg(1-\sum_{T\in\mathcal{T}\colon v\in V[T]} \mu_T\Bigg)\cdot \pi_v\\
&
\leq
\bar c^\top x + \sum_{v\in V} (1-y^*_v)\pi_v \\
&\leq
2c^\top x^* + \pi^\top (1-y^*)\enspace.
\end{align*}
Clearly, the same bound holds for a minimum spanning tree on $V[T_0]$.
Thus again, we get the desired bound in expectation, hence for at least one $T\in\mathcal{T}$, we obtain the claimed guarantee.
\end{proof}

\section{From splitting off to tree decompositions}\label{sec:splitting}

In this section, we prove \cref{lem:PCTSP_sol_properties,lem:solution_decomposition,lem:fractional_tree_partition}.
All three of these lemmas share a common technique in the background, namely \emph{splitting off}, hence we start by recalling the concept.
Splitting off was initially introduced as a fundamental tool in Graph Theory \cite{lovasz_1976_connectivity,mader_1978_reduction, frank_1992_on} for modifying a graph while maintaining certain connectivity properties.
We exploit a weighted version of splitting off.
Here, given a complete graph $G=(V,E)$ and edge weights $w\colon E \to \mathbb{R}_{\ge 0}$, a splitting operation at a vertex $v\in V$ is the following:
For two edges $e=\{v,u\}$ and $f=\{v,w\}$ incident to $v$ and $\delta \in (0, \min\{w(e),w(f)\}]$, reduce the edge weights of $e$ and $f$ by $\delta$, and increase the edge weight of $\{u,w\}$ by $\delta$.
In case $e=f=\{v,u\}$, the weight of $e$ is reduced by $\delta$, and we call the splitting operation \emph{degenerate}.
A splitting $(e,f,\delta)$ at $v$ is called \emph{feasible} if minimum $s$-$t$ cut sizes in $G$ with respect to the weights $w$ are preserved under the splitting for all $s,t\in V\setminus\{v\}$.
A \emph{complete splitting} at $v \in V$ denotes a sequence of feasible splitting operations at $v$ such that all edges incident to $v$ have weight zero after performing the splittings.
\textcite{frank_1992_on} showed that such a complete splitting always exists. It is straightforward to see that a complete splitting consisting of less than $|V|^2$ splitting operations can be found through a polynomial number of minimum $s$-$t$ cut computations, implying the following theorem.

\begin{theorem}\label{thm:complete_splitting}
Let $G=(V,E)$ be a complete graph with edge weights $w\colon E \to \mathbb{R}_{\ge 0}$.
Let $v \in V$ and $\beta \in [0,w(\delta(v))]$.
There is a deterministic algorithm that computes in polynomial time a sequence of less than $\operatorname{poly}(|V|)$ many feasible splitting operations at $v$ such that the edges incident to $v$ have total weight $\beta$ with respect to the resulting weight function.
Note that if $\beta = 0$, this results in a complete splitting at $v$.
\end{theorem}

We remark that it is known that already a linear number of splitting operations are sufficient to obtain a complete splitting.
For further details, we refer to \cite{complete_splitting, nagamochi_2006_fast} and references therein.
Moreover, results of the above kind are typically proved for $\beta=0$ only.
The extension to $\beta>0$ is immediate, since a sequence of splitting operations at a vertex $v$ can simply be pruned as soon as the edges incident to $v$ have remaining weight $\beta$.
Having \cref{thm:complete_splitting} at hand, we can readily prove \cref{lem:PCTSP_sol_properties}.

\begin{proof}[Proof of \cref{lem:PCTSP_sol_properties}]
Consider the point $(\lambda x, y_\lambda)$.
This point satisfies all cut constraints in the \ref{eq:rel_PCTSP}, but might violate some of the degree constraints.
If not, we are done; hence assume the contrary and fix a vertex $s\in V$ such that the degree constraint at $s$ is violated.
We then have $\lambda x(\delta(s)) = 2\lambda y_s > 2y_{\lambda,s} = 2$.
Thus, we can apply \cref{thm:complete_splitting} with $\beta = 2$ at the vertex $s$, resulting in a new weight function $\bar x$ after performing the splitting.
All splitting operations are feasible, hence minimum $r$-$v$ cut sizes are preserved for all $v \in V \setminus \{r,s\}$, and, in particular, equal to $\bar x(\delta(v))=\lambda x(\delta(v))$ due to the initial degree constraints satisfied by $x$.
Consequently, $(\bar x, y_\lambda)$ satisfies all cut constraints in the \ref{eq:rel_PCTSP}.
By the same argument, there cannot be degenerate splitting operations.
Hence, degrees of vertices other than $s$ are invariant under the splitting, i.e., degree constraints satisfied by $\lambda x$ are also satisfied by $\bar x$.
Additionally, $\bar x$ satisfies the degree constraint at $s$ by construction.
Iterating this for all vertices $s\in V$ with violated degree constraints leads to a pair $(x_\lambda, y_\lambda)$ feasible for the \ref{eq:rel_PCTSP}.
By the triangle inequality, applying a splitting off operation does not increase the total cost of the solution, hence $c^\top x_\lambda \leq \lambda c^\top x$, as desired.
\end{proof}

Next, we show \cref{lem:solution_decomposition,lem:fractional_tree_partition}.
More precisely, we prove that the following common generalization of the two lemmas is true.

\begin{lemma}\label{lem:general_solution_decomposition}
Let $(x, y)$ be feasible for the \ref{eq:rel_PCTSP}, and let $U \subseteq V_1 \coloneqq \{v\in V\colon y_v = 1\}$. Assume that there is an edge $e_0 \in \binom{U}{2}$ with $x_{e_0}\ge1$.
We can in polynomial time construct a set $\mathcal{T}$ of trees and weights $\mu \in [0,1]^\mathcal{T}$ with the following properties:
\begin{enumerate}

\item\label{item:general_solution_partitioning} The solution $x$ is a conic combination of the trees in $\mathcal{T}$ with weights $\mu$ and the edge $e_0$, i.e.,
$$
x = \sum_{T\in \mathcal{T}}\mu_T\chi^{E[T]} + \chi^{e_0} \enspace.
$$

\item\label{item:general_incident_trees} For every $v\in V \setminus U$,
$$
\sum_{T\in \mathcal{T}\colon v\in V[T]} \mu_T = y_v \enspace.
$$

\item\label{item:general_anchor_points} For every $T\in\mathcal{T}$, we have $|V[T]\cap U|=2$, and we call the vertices in $V[T]\cap U$ the \emph{anchors} of $T$.

\item\label{item:general_subtour_property} For $T \in \mathcal{T}$, let $e_T\coloneqq V[T]\cap U$ denote the edge joining the anchors of $T$, and let $H\coloneqq (U, F)$ be the multigraph with edge set $F\coloneqq \{e_0\}\cup \{e_T\colon T\in \mathcal{T}\}$.
Then
$$
z \coloneqq \sum_{T\in \mathcal{T}}\mu_T\chi^{e_T} + \chi^{e_0} \in \PHK(H)\enspace.
$$
\end{enumerate}

\end{lemma}

We start by showing that indeed, \cref{lem:general_solution_decomposition} is a common generalization of \cref{lem:solution_decomposition,lem:fractional_tree_partition}.

\begin{proof}[Proof of \cref{lem:solution_decomposition,lem:fractional_tree_partition}]
\cref{lem:solution_decomposition} immediately follows from \cref{lem:general_solution_decomposition} by choosing $U=V_1$.
In order to utilize \cref{lem:general_solution_decomposition} to prove \cref{lem:fractional_tree_partition}, we note that we can guarantee that \cref{ass:e_0} is satisfied:
If not, modify the given instance as described in \cref{sec:our_approach} and transform the trees computed in the thereby obtained auxiliary graph into trees in the original graph by contracting $e_0$ and deleting some edge if necessary.

Now, choose $U$ to be the vertex set containing the endpoints of $e_0$ only.
In this case, the graph $H$ in \cref{item:general_subtour_property} consists of parallel edges only, with one of them being $e_0$, and every tree in $\mathcal{T}$ corresponding to one of the other edges.
Moreover, the total edge weight of any feasible solution in the Held-Karp polytope over a two-vertex graph is $2$.
In $z$, at least one unit is taken by the edge $e_0$; hence we must have $\sum_{T\in\mathcal{T}}\mu_T \le 1$. On the other hand, by \cref{item:general_incident_trees} $$\sum_{T\in\mathcal{T}}\mu_T \ge \sum_{T \in \mathcal{T}\colon r\in V[T]} \mu_T = y_r = 1 \enspace,$$ implying $\sum_{T\in\mathcal{T}}\mu_T = 1$, and $r \in V[T]$ for each $T \in \mathcal{T}$.
\end{proof}

It remains to prove \cref{lem:general_solution_decomposition}.
On a high level, the proof is organized as follows.
We apply \cref{thm:complete_splitting} iteratively to find complete splittings at vertices $v\in V\setminus U$, in order of increasing connectivity $y_v$, until only vertices in $U$ remain.
In the resulting graph, it is trivial to find a decomposition meeting the requirements: We can simply consider each individual edge as a tree, with the corresponding weight being the edge weight.
Carefully undoing the splitting off operations in a way similar to \textcite{bang-jensen_1995_preserving}, and exploiting that we did them in increasing order of connectivity, we will show how to alter the initial set of trees so to obtain one with the desired properties for the original instance.

\begin{proof}[Proof of \cref{lem:general_solution_decomposition}]
We prove the statement by induction on $|V \setminus U|$.
If $|V \setminus U|=0$, there is nothing to be done: Creating a tree $T$ from every single edge, with $\mu_T$ being equal to the edge weight (or equal to the edge weight minus 1 in case of $e_0$), clearly satisfies all the properties.

If $|V\setminus U|>0$, consider a vertex $s\in V\setminus U$ of minimum connectivity $y_s$.
By \cref{thm:complete_splitting}, we can efficiently compute a sequence of feasible splitting off operations resulting in a complete splitting at $s$. Note that the weights of minimum $r$-$v$-cuts are preserved under the splittings for all $v \in V \setminus \{r,s\}$.
By the degree constraints of the \ref{eq:rel_PCTSP}, this implies that the degrees of vertices different from $s$ are preserved as well.
The latter has two implications:
First, none of the splitting operations at $s$ can be degenerate, as degenerate splittings would reduce the degree of some vertex.
Second, the graph on $V'=V \setminus \{s\}$ satisfies the assumptions of \cref{lem:general_solution_decomposition} with respect to the resulting edge weights.
Consequently, by the inductive assumption, we can compute in polynomial time edge-disjoint trees $\mathcal{T}$ with the desired properties; in particular, for every $v\in V' \setminus U$,
$$
\sum_{T\in \mathcal{T}\colon v\in V[T]} \mu_T = y_v \enspace.
$$

We now undo the splitting operations at $s$ and modify the trees in $\mathcal{T}$ accordingly.
We note that while intermediate steps may lead to $\mathcal{T}$ being a multiset of trees, we can merge identical trees at the very end by simply adding their weights.
By \cref{item:subtour_property}, this does not generate weights larger than $1$.
Before we start undoing the splitting operations, we initialize auxiliary variables $\spare{v}=0$ for each $v \in V'$.

Let $(e=\{s,u\},f=\{s,v\},\delta)$ with $e \neq f$ be one of the splitting off operations at $s$ that we want to revert.
Let $\mathcal{T}'=\{T_1,\dots, T_k\} \subseteq \mathcal{T}$ be a minimal subset of trees with $\{u,v\} \in E[T]$ for each $T \in \mathcal{T}'$ and $\sum_{T\in \mathcal{T}'} \mu_T \ge \delta$.
Let $\varepsilon =  \sum_{T\in \mathcal{T}'}\mu_T -\delta$.
Note that $\varepsilon < \mu_{T_k}$ by the minimality of $\mathcal{T}'$.
If $\varepsilon > 0$, add a copy $T_k'$ of $T_k$ to $\mathcal{T}$, set $\mu_{T'_k} = \varepsilon$, and reduce $\mu_{T_k}$ by $\varepsilon$.
After this modification, $\sum_{T \in \mathcal{T}'} \mu_T= \delta$.
For each $T \in \mathcal{T}'$ do the following:
\begin{enumerate}
\item\label{item:not_contained} If $s \notin V[T]$, remove $\{u,v\}$ from $E[T]$ and add $\{s,u\}$ and $\{s,v\}$ to $E[T]$.

\item\label{item:contained} If $s \in V[T]$, remove $\{u,v\}$ from $E[T]$ and add either $\{s,u\}$ or $\{s,v\}$ to $E[T]$ such that $T$ remains acyclic.
If $\{s,u\}$ is added to $E[T]$ increase $\spare{v}$ by $\mu_T$. Otherwise, increase $\spare{u}$ by $\mu_T$.
\end{enumerate}
Note that in case \ref{item:not_contained}, the total weight of trees containing $s$ increases by $\mu_T$.
Otherwise, either $\spare{u}$ or $\spare{v}$ increases by $\mu_T$.
Consequently, through the above operations, $\sum_{T \in \mathcal{T}\colon s \in V[T]} \mu_T + \sum_{v \in V'} \spare{v}$ increases by $\delta$. Since the splitting off operation $(e,f,\delta)$ decreased the degree of $s$ by $2\delta$, we get, after reverting all splitting off operations at $s$,
\[
\sum_{T \in \mathcal{T}\colon s \in V[T]} \mu_T + \sum_{v \in V'} \spare{v} = \frac{x(\delta(s))}{2} = y_s\enspace.
\]

Now, for each $w \in V'$ do the following:
If $\spare{w} > 0$, find a minimal subset of trees $\mathcal{T}'' = \{T_1,\dots, T_l\} \subseteq \mathcal{T}$ with $w \in V[T]$ but $s \notin V[T]$ for each $T \in \mathcal{T}''$, and $\sum_{T\in \mathcal{T}''} \mu_T \ge \spare{w}$.
Note that such a subset of trees always exists since
\[
\sum_{T\in \mathcal{T}\colon w\in V[T]} \mu_T \ge y_w \ge y_s = \sum_{T \in \mathcal{T}\colon s \in V[T]} \mu_T + \sum_{v \in V'} \spare{v} \enspace,
\]
where we used in the second inequality that $s$ has minimum connectivity. Note that the first inequality is satisfied by the inductive assumption:
If $w \notin U$, the bound is tight by \cref{item:general_incident_trees} of \cref{lem:general_solution_decomposition}.
If $w \in U$, then $\sum_{T\in \mathcal{T}\colon w\in V[T]} \mu_T \geq 1 = y_w$ by \cref{item:general_subtour_property} of \cref{lem:general_solution_decomposition} and the degree constraints in the Held-Karp polytope (note that $e_0$ may be incident to $w$).
Let $\varepsilon = \sum_{T\in \mathcal{T}''}\mu_T - \spare{w}$. Note that $\varepsilon < \mu_{T_l}$ by the minimality of $\mathcal{T}''$.
If $\varepsilon > 0$, add a copy $T_l'$ of $T_l$ to $\mathcal{T}$ and set $\mu_{T'_l} = \varepsilon$ and reduce $\mu_{T_l}$ by $\varepsilon$.
For each $T \in \mathcal{T}''$, add $\{s,w\}$ to $E[T]$.
Note that this increases the total weight of the trees containing $s$ by $\spare{w}$.
Hence, after using up all spares, we get
\begin{equation}\label{eq:incident_trees}
\sum_{T \in \mathcal{T}\colon s \in V[T]} \mu_T = y_s\enspace.
\end{equation}

Note that throughout the above operations, the graphs $T \in \mathcal{T}$ are trees, and after reverting all splitting off operations, $x$ is a conic combination of the trees in $\mathcal{T}$ with weights $\mu$ and the edge $e_0$ by construction.
Furthermore, the above operations do not change the intersection of the trees in $\mathcal{T}$ with $U$, and splitting trees does not affect \cref{item:general_subtour_property}.
Hence, \cref{item:general_anchor_points} and \cref{item:general_subtour_property} still hold after the above operations.
Moreover, for every vertex in $V'$, the total weight of trees covering this vertex is unchanged.
Hence, \cref{item:general_incident_trees} is still satisfied for each $v \in V' \setminus U$, and also for $s$ by~\eqref{eq:incident_trees}.

Finally, we note that our construction can be executed in polynomial time.
Indeed, by \cref{thm:complete_splitting}, in each step of our inductive procedure, we have to revert less than $\operatorname{poly}(|V|)$ many splitting operations, which increases the total number of trees in $\mathcal{T}$ by an additive $\operatorname{poly}(|V|)$.
This implies that the size of $\mathcal{T}$ remains polynomially bounded throughout.
\end{proof}
\section{Sampling walks and derandomization}\label{sec:randomization}

In this section, we prove \cref{lem:walk_sampling}, which is at the core of our new randomized algorithm, giving the basis of the returned tour.
Furthermore, we show how to replace the randomized steps of our algorithm with deterministic procedures while obtaining the same guarantees.
This includes derandomizing \cref{lem:walk_sampling}, but also derandomizing the random choice of a suitable threshold $\gamma$.
Together with the analysis presented earlier, these ingredients prove our main result, \cref{thm:improvedApproximation}.

Before going into the more technical parts, let us recall the pipage rounding procedure, which we use in a black-box way for both the proof of \cref{lem:walk_sampling} and its deterministic analogue.
Pipage rounding is a technique that goes back to \textcite{ageev_2004_pipage,srinivasan_2011_distributions}.
Generalized to a matroid setting and put into a randomized framework by \textcite{calinescu_2011_maximizing},
randomized pipage rounding is one of several extensively studied efficient sampling procedures in matroid and matroid base polytopes (also see \cite{chekuri_2010_dependent,harvey_2014_pipage} and references therein).
Generally, the idea is to start from a fractional point $x$, and modify this solution iteratively in a well-chosen way until an integral solution is found.
In pipage rounding, each single modification is a maximal step along a direction parallel to some $e_i-e_j$, where $e_k$ is the $k^{\text{th}}$ unit vector.
Because the directions $e_i-e_j$ are precisely the edge directions of matroid base polytopes, such moves are enough to eventually end up in a vertex.
In a randomized setting, probabilities for going either way are chosen such that in every step, marginals of every component are preserved.
Concretely, if for a matroid $\mathcal{M}$, $B_{\mathcal{M}}$ denotes the associated matroid base polytope, i.e., the convex hull of all characteristic vectors of bases in $\mathcal{M}$, one can obtain the following.

\begin{theorem}[{\cite[direct implication of Theorem 1.1]{chekuri_2010_dependent}}]\label{thm:matroidRounding}
Let $x=(x_1,\ldots,x_n)\in B_{\mathcal{M}}$ be a fractional solution in the matroid base polytope, and let $X=(X_1,\ldots,X_n)\in\{0,1\}^n$ be an integral solution obtained from $x$ using randomized pipage rounding.
Then, the following holds.
\begin{enumerate}
\item\label{item:marginals} For every $i\in [n]$, we have $\Exp[X_i]=x_i$.
\item\label{item:neg_correlation} For every $T\subseteq [n]$, we have
$$
\Exp\left[\prod_{i\in T}X_i\right]\leq \prod_{i\in T}x_i
\qquad\text{and}\qquad
\Exp\left[\prod_{i\in T}(1-X_i)\right]\leq \prod_{i\in T}(1-x_i)\enspace.
$$
\end{enumerate}
\end{theorem}

Thus, besides preserving marginals of a point in the spanning tree polytope exactly, randomized pipage rounding also guarantees that the events of edges being sampled as well as those of edges not being sampled are negatively correlated, a property that we crucially exploit to bound the total penalty that our \PCTSP{} solutions incur.
We remark that there are other known procedures achieving the same guarantees, in particular randomized swap rounding \cite{chekuri_2010_dependent}, or, at a $1\pm\varepsilon$ error in the marginals, sampling from $\lambda$-uniform spanning tree distributions (see, for example, \cite{asadpour_2017_lognloglogn, karlin_2021_slightly}).

For a deterministic version of pipage rounding, we use a formulation of \textcite{harvey_2014_pipage} that is based on earlier work mentioned above.
Here, the decisions on the direction of every single step are guided by a function $g\colon \mathcal{B}_{\mathcal{M}}\to\mathbb{R}$ that is required to be \emph{concave under swaps}, i.e., for all $x\in\mathcal{B}_{\mathcal{M}}$ and all unit vectors $e_i$, $e_j$, we require that the function $t\mapsto g(x+t(e_i-e_j))$ is concave.
This concavity under swaps allows to perform steps along decreasing values of $g$, and one can obtain the following.

\begin{theorem}[Deterministic Pipage Rounding \cite{harvey_2014_pipage}]\label{thm:det_pipage}
There is a deterministic, polynomial-time algorithm that, given $x \in B_{\mathcal{M}}$ and a value oracle for a function $g$ that is concave under swaps, outputs an extreme point $\hat{x}$ of $B_{\mathcal{M}}$ with $g(\hat{x}) \le g(x)$.
\end{theorem}

\subsection[Proof of Lemma 7]{Proof of \cref{lem:walk_sampling}}\label{sec:proofWalkSampling}

As already outlined earlier, we prove \cref{lem:walk_sampling} by exploiting the decomposition of a solution $(x,y)$ of the \ref{eq:rel_PCTSP} into a family $\mathcal{T}$ of trees through \cref{lem:solution_decomposition}.
Concretely, we will split each tree into two walks with corresponding weights, one of them consisting of backbone edges only, while the other contains two copies of the limb edges on top (see \cref{fig:backbones_and_limbs} on \cpageref{fig:backbones_and_limbs}).
To sample walks, we leverage \cref{item:subtour_property} of \cref{lem:solution_decomposition}, concretely the implication that $z-\chi^{e_0}$ is in the spanning tree polytope over the multigraph $G_1=(V_1, E_1)$, where $E_1 \coloneqq \{e_0\} \cup \{e_T\colon T \in \mathcal{T}\}$.
In this graph, every edge corresponds to a tree; by duplicating every edge and splitting the edge weight accordingly, we can associate to every edge one of the walks constructed earlier.
Also, the resulting edge weights still form a point in the spanning tree polytope of the blown-up graph, and we sample a spanning tree in this graph with the edge weights as marginals through pipage rounding.
Sampled edges can then be replaced by the corresponding walks. The union of the edge sets of all these walks induces the multigraph $H$ claimed by \cref{lem:walk_sampling}.
Using the properties guaranteed by \cref{thm:matroidRounding}, we can prove \cref{lem:walk_sampling}.

\begin{proof}[Proof of \cref{lem:walk_sampling}]
Let $\mathcal{T}$ and weights $\mu\in[0,1]^\mathcal{T}$ be obtained from \cref{lem:solution_decomposition}.
We construct a set $\walks$ of walks as follows.
For each $T \in \mathcal{T}$ with anchors $s,t\in V_1$, add one $s$-$t$ walk $W_1(T)$ to $\walks$ that consists of $\bb{T}$ and two copies of $\lb{T}$, and set $\nu_{W_1(T)} = \frac34 \mu_T$.
As we add two copies of the limb edges, the resulting edge set has even degrees except at $s$ and $t$, and can thus indeed be cast as an $s$-$t$ walk.
On top of that, add another $s$-$t$ walk $W_2(T)$ to $\walks$ that consists of $\bb{T}$ only, and set $\nu_{W_2(T)} = \frac14 \mu_T$.

For each walk $W \in \walks$, let $e_W$ be the direct edge joining its start- and endpoint.
Let $F$ be the multiset of all these edges.
Now, consider the point $z_0 \in [0,1]^{F}$ defined by $z_{0}(e_W) = \nu_W$ for each $W \in \walks$.
As $z_0$ and the point $z-\chi^{e_0}$ (where $z$ is defined as in \cref{item:subtour_property} of \cref{lem:solution_decomposition}) differ only by splitting edges into parallel edges of the same total weight, it immediately follows that $z_0+\chi^{e_0}\in\PHK(G_0)$, where $G_0=(V_1,F)$.
From the polyhedral descriptions of $\PHK(G_0)$ and the spanning tree polytope $\PST(G_0)$ over $G_0$, one can see that this in turn implies that $z_0\in\PST(G_0)$.
Applying \cref{thm:matroidRounding} with $z_0$, we thus obtain a spanning tree of $(V_1,F_1)$ with $F_1\subseteq F$.
Let
$$
\mathcal{W} \coloneqq \{W \in \walks\colon e_W \in F_1\}\enspace.
$$
Because $\mathcal{W} \subseteq \walks$, each walk in $\mathcal{W}$ starts and ends in $V_1$.
Since $(V_1,F_1)$ is spanning, the multigraph $H$ induced by the edge sets of the walks in $\mathcal{W}$ is connected and always spans $V_1$. Let $v \in V \setminus V_1$. Then
\begin{align*}
	\tag{by \cref{thm:matroidRounding}\,\ref{item:neg_correlation}}
    \Prob[v \notin H]
    &\leq
    \prod_{W \in \walks \colon v \in V[W]} \Prob[e_W \notin F] \\
    &\le \prod_{T \in \mathcal{T}\colon v \in V[T]} \Prob[e_{W_1(T)} \notin F] \\
	\tag{by \cref{thm:matroidRounding}\,\ref{item:marginals}}
	&= \prod_{T \in \mathcal{T}\colon v \in V[T]} \left( 1 - \frac34 \mu_T \right)
    \\
    \tag{using $1-x\leq \exp(-x)$}
    &\leq
    \exp\Bigg(-\frac34\sum_{T\in \mathcal{T}\colon v\in V[T]} \mu_T\Bigg)\\
    \tag{by  \cref{lem:solution_decomposition}}
    &=\exp\left(-\frac{3y_v}4\right)\enspace.
\end{align*}
Consequently,
\begin{equation}\label{eq:exp_penalty}
\Exp[\pi(V\setminus V[H])]
= \sum_{v \in V \setminus V_1} \pi_v \Pr[v\notin H]
\leq \sum_{v \in V \setminus V_1} \pi_v \exp\left(-\frac{3y_v}4\right)\enspace.
\end{equation}
It remains to bound the expected total cost of the sampled walks. Using \cref{thm:matroidRounding}\,\ref{item:marginals} again, we get
\begin{align}
\nonumber
	\Exp[c(E[H])] &= \sum_{W\in \walks} \nu_W \cdot c(E[W]) \\
\nonumber
	&= \sum_{T\in \mathcal{T}} \left( \frac34 \mu_T \cdot c(E[W_1(T)]) + \frac14 \mu_T \cdot c(E[W_2(T)])  \right) \\
\label{eq:exp_length}
	&= \sum_{T\in \mathcal{T}} \mu_T \left( c(\bb{T}) + \frac32 \cdot c(\lb{T}) \right)\enspace.
\end{align}
Together, \eqref{eq:exp_penalty} and \eqref{eq:exp_length} imply the bound claimed by \cref{lem:walk_sampling}.
Finally, note that \cref{thm:matroidRounding} guarantees that we can sample $\mathcal{W}$, and thus obtain $H$, efficiently.
\end{proof}

\subsection{A deterministic selection of walks}\label{sec:derandomizeWalkSampling}

In this section, we show how to deterministically construct the graph $H$ given by \cref{lem:walk_sampling} while maintaining the overall guarantees on the performance of our algorithm.
To be precise, we show the following deterministic analogue of \cref{lem:walk_sampling}.

\begin{lemma}\label{lem:deterministic_walks}
Let $(x, y)$ be a feasible solution of the \ref{eq:rel_PCTSP} satisfying \cref{ass:e_0}.
Let $\mathcal{T}$ and $\mu\in [0,1]^\mathcal{T}$ be a set of trees and weights with the properties guaranteed by \cref{lem:solution_decomposition}.
We can deterministically and in polynomial time obtain a multigraph $H$ on a subset of $V$ with the following properties:
\begin{enumerate}
\item $H$ is connected, spans $V_1\coloneqq \{v\in V\colon y_v=1\}$, and all vertices in $V\setminus V_1$ have even degree.

\item\label{item:det_H_bound} We have
\begin{multline}
c(E[H]) + \pi(V\setminus V[H]) \\
\leq \sum_{T \in \mathcal{T}} \mu_T \cdot \left( c(\bb{T}) + \frac32 \cdot c(\lb{T}) \right) + \sum_{v\in V\setminus V_1}\pi_v \exp\left(-\frac{3y_v}{4}\right)\enspace.
\end{multline}
\end{enumerate}
\end{lemma}

To prove \cref{lem:deterministic_walks}, we recall the randomized construction from the proof of \cref{lem:walk_sampling} (see \cref{sec:proofWalkSampling}):
There, the graph $H$ is based on a spanning tree sampled in an auxiliary graph through randomized pipage rounding.
We now show that the deterministic analogue given through \cref{thm:det_pipage} can achieve the same guarantees.

\begin{proof}[Proof of \cref{lem:deterministic_walks}]
Starting from a solution $(x,y)$ of the \ref{eq:rel_PCTSP} and repeating the construction from the proof of \cref{lem:walk_sampling}, we define the multigraph $G_0=(V_1,F)$, where for each $T\in\mathcal{T}$, there are two parallel edges in $F$ joining the anchors of $T$:
One corresponding to the walk $W_1(T)$ consisting of $\bb{T}$ and two copies of $\lb{T}$ with weight $\nu_{W_1(T)}=\frac34\mu_{T}$, and one corresponding to the walk $W_2(T)$ consisting of $\bb{T}$ with weight $\nu_{W_2(T)}=\frac14\mu_T$.
Again, we denote by $\walks$ the set of all walks constructed this way, and let $e_W$ be the edge in $F$ corresponding to the walk $W$. Moreover, we know that $z_0\in[0,1]^F$ defined by $z_0(e_W)=\nu_W$ satisfies $z_0\in\PST(G_0)$.

We want to apply deterministic pipage rounding, i.e., \cref{thm:det_pipage}, in the spanning tree polytope $\PST(G_0)$ starting from $z_0$ with the function $g\colon \PST(G_0)\to\mathbb{R}$ defined by
\[
g(z) \coloneqq \sum_{W \in \walks} c(E[W]) z_{e_W}  + \sum_{v \in V \setminus V_1} \pi_v \prod_{W \in \walks\colon v \in W} (1- z_{e_W}) \enspace.
\]
To this end, we prove the following properties.
\begin{claim}\label{claim:whatever}\leavevmode
\begin{enumerate}
\item\label{item:concavity} $g$ is concave under swaps.
\item\label{item:integralEvaluation} If $z$ is the incidence vector of a spanning tree in $G_0$, and $H$ is the multigraph obtained from that spanning tree by replacing every edge $e_W$ by the walk $W$, then $g(z)=c(E[H])+\pi(V\setminus V[H])$.
\item\label{item:crazyBound} We have
$$
g(z_0)\leq \sum_{T \in \mathcal{T}} \mu_T \cdot \left( c(\bb{T}) + \frac32 \cdot c(\lb{T}) \right) + \sum_{v\in V\setminus V_1}\pi_v \exp\left(-\frac{3y_v}{4}\right)\enspace.
$$
\end{enumerate}
\end{claim}

Clearly, \cref{lem:deterministic_walks} immediately follows from the claim: \cref{item:concavity} allows for applying \cref{thm:det_pipage}, while \cref{item:integralEvaluation,item:crazyBound} show that the graph $H$ satisfies the bound in \cref{item:det_H_bound} of \cref{lem:deterministic_walks}.
Moreover, $H$ is clearly connected, spans $V_1$, and has even degrees in $V\setminus V_1$ because it was obtained from a spanning tree on $V_1$ through replacing edges by walks connecting the same endpoints.
Consequently, we are left with proving the claim.

\smallskip\noindent{\itshape Proof of \cref{item:concavity}.}
The first sum in the definition of $g(z)$ is linear in $z$, hence it suffices to show that $\bar g\colon \PST(G_0)\to\mathbb{R}$ defined by
\[
\bar g(z) \coloneqq \sum_{v \in V \setminus V_1} \pi_v \prod_{W \in \walks\colon v \in W} (1- z_{e_W})
\]
is concave under swaps.
We show that $\bar g$ is in fact equal to (a restriction of) the multilinear extension of a supermodular set function $f\colon 2^{\walks}\to\mathbb{R}$, which readily implies that $\bar g$ is concave under swaps  \cite{calinescu_2011_maximizing}.
For $S \subseteq \walks$, let $V[S]$ denote the union of all vertices covered by the walks in $S$, and define the set function $f\colon 2^{\walks}\to\mathbb{R}$ by
\[
f(S) \coloneqq \pi(V \setminus V[S]) = \sum_{v\in V} \pi_v \cdot\left(1-\mathds{1}_{v\in V[S]}\right)\enspace.
\]
Supermodularity follows immediately from the sum representation above, because the indicator function $\mathds{1}_{v\in V[S]}$ is clearly submodular.
By definition, the multilinear extension $F\colon [0,1]^{\walks} \rightarrow \mathbb{R}$ of $f$ is given by
\[
	F(z) = \sum_{S \subseteq \walks} f(S) \prod_{W \in S} z_W \prod_{W \in \walks \setminus S} (1- z_W) \enspace .
\]
It remains to show that $F(z)=\bar g(z)$. Indeed,
\begin{align*}
F(z)
&=
\sum_{S \subseteq \walks} \pi(V \setminus V[S]) \prod_{W \in S\vphantom{\setminus}} z_W \prod_{W \in \walks \setminus S} (1- z_W) \\
&= \sum_{v \in V \setminus V_1} \pi_v \sum_{S \subseteq \{W \in \walks\colon v \notin W\}} \prod_{W \in S\vphantom{\setminus}} z_W \prod_{W \in \walks \setminus S} (1- z_W) \\
&= \sum_{v \in V \setminus V_1} \pi_v \prod_{W \in \walks\colon v \in W\vphantom{\setminus}} (1- z_W)\underbrace{\sum_{S \subseteq \{W \in \walks\colon v \notin W\}} \prod_{W \in S\vphantom{\setminus}} z_W \prod_{W \in  \{W \in \walks\colon v \notin W\} \setminus S} (1- z_W)}_{=1}
= \bar g(z) \enspace.
\end{align*}

\smallskip\noindent{\itshape Proof of \cref{item:integralEvaluation}.}
If $z$ is integral and corresponds to a spanning tree $T$ of $G_0$, we may write
$$
g(z) = \sum_{e_W\in T}c(E[W]) + \sum_{v\in V\setminus V_1}\pi_v\prod_{e_W\in E[G_0]\colon v\in W} \mathds{1}_{e_W\notin E[T]} \enspace.
$$
The first term equals $c(E[H])$ by definition of $H$.
In the second term, the sum has nonzero terms only for $v\in V\setminus V_1$ that are not covered by any walk $W$ for all $e_W\in E[T]$, i.e., precisely for the vertices not covered by $H$.
As $H$ always covers $V_1$, the second term thus equals $\pi(V\setminus V[H])$, as desired.

\smallskip\noindent{\itshape Proof of \cref{item:crazyBound}.}
By definition of the walks in $\walks$, we have
\begin{align*}
\sum_{W \in \walks} c(E[W]) z_{0}(e_W) &= \sum_{T\in\mathcal{T}} \left( c(E[W_1(T)])\cdot\frac34\mu_T + c(E[W_2(T)])\cdot\frac14 \mu_T\right) \\
&= \sum_{T \in \mathcal{T}} \mu_T \cdot \left( c(\bb{T}) + \frac32 \cdot c(\lb{T}) \right)\enspace,
\end{align*}
and
\begin{align*}
\sum_{v \in V \setminus V_1} \pi_v \prod_{W \in \walks\colon v \in W} (1- z_{0}(e_W))
&\leq \sum_{v \in V \setminus V_1} \pi_v \prod_{T \in \mathcal{T}\colon v \in V[T]} \left(1- \frac34\mu_T\right)\\
&\leq \sum_{v \in V \setminus V_1} \pi_v \exp\Bigg(-\frac34\sum_{T \in \mathcal{T}\colon r\in V[T]}\mu_T\Bigg) = \sum_{v \in V \setminus V_1} \pi_v \exp\left(-\frac34y_v\right)\enspace,
\end{align*}
which together gives the claimed relation, finishes the proof of \cref{claim:whatever}, and thus also of \cref{lem:deterministic_walks}.
\end{proof}

\subsection{Deterministic selection of a threshold}\label{sec:derandomizeGamma}

In the analysis of \cref{alg:treeSampling} in \cref{thm:randomizedAnalysis}, we also exploited a random choice of the threshold parameter $\gamma$.
We now show that if we start from an optimal solution $(x^*,y^*)$ of the \ref{eq:rel_PCTSP}, it suffices to try all values in the set $\{y^*_v\colon v\in V\}$ as thresholds.
\begin{lemma}\label{lem:choiceOfGamma}
	Let $(x^*, y^*)$ be an optimal solution of the \ref{eq:rel_PCTSP}, and let $\alpha>0$.
	If there exists $\gamma\in(0,1]$ such that
	\begin{equation}\label{eq:goodGamma}
		\frac{3}{2\gamma} \cdot c^\top x^* + \sum_{v \in V\colon y^*_v < \gamma} \exp\left({-\frac{3}{4\gamma}y_v^*}\right) \cdot \pi_v
		\leq
		\alpha\cdot\left( c^\top x^* + \pi^\top (1-y^*)\right)\enspace,
	\end{equation}
	then there exists a vertex $v\in V$ such that the above also holds for $\gamma = y_v^*$.
\end{lemma}

\begin{proof}[Proof of \cref{lem:choiceOfGamma}]
Fix $\gamma\in(0,1]$ such that \eqref{eq:goodGamma} holds.
We claim that with $y\coloneqq \min\{y_v^*\colon y_v^*\geq \gamma\}$, we have
\begin{equation}\label{eq:gamma_monotonicity}
\frac{3}{2y} \cdot c^\top x^* + \sum_{v \in V\colon y^*_v < y} \exp\left({-\frac{3}{4y}y_v^*}\right) \cdot \pi_v
\leq
\frac{3}{2\gamma} \cdot c^\top x^* + \sum_{v \in V\colon y^*_v < \gamma} \exp\left({-\frac{3}{4\gamma}y_v^*}\right) \cdot \pi_v\enspace,
\end{equation}
which obviously implies the lemma.
For the sake of deriving a contradiction, assume that \eqref{eq:gamma_monotonicity} does not hold, and note that because the sums in \eqref{eq:gamma_monotonicity} are over the same subset of vertices by definition of $y$, this assumption can be rewritten as
$$
\frac{3}{2}\left(\frac1\gamma \cdot c^\top x^* - \frac{1}{y} \cdot c^\top x^*\right)
<
\sum_{v \in V\colon y^*_v < y} \left(\exp\left({-\frac{3y_v^*}{4y}}\right)-\exp\left({-\frac{3y_v^*}{4\gamma}}\right)\right)\cdot\pi_v\enspace.
$$
As $x\mapsto \exp(-x) + x$ is increasing for $x\geq 0$, we have $\exp(-a)-\exp(-b)\leq b-a$ whenever $b\geq a\geq0$, hence
$$
\frac{3}{2}\left(\frac1\gamma \cdot c^\top x^* - \frac{1}{y} \cdot c^\top x^*\right)
<
\sum_{v \in V\colon y^*_v < y} \frac34\left(\frac{y_v^*}{\gamma} - \frac{y_v^*}{y} \right) \pi_v
=
\frac34\left(\pi^\top (1-y^*_{\sfrac1y})-\pi^\top(1-y^*_{\sfrac1\gamma})\right)\enspace,
$$
where we use the notation $y_\lambda^*$ of \cref{lem:PCTSP_sol_properties} for $\lambda=\sfrac1y$ and $\lambda=\sfrac1\gamma$.
Through scaling down the non-negative left-hand side by $\frac12$ and rearranging terms, this further implies
\begin{align}\label{eq:non-increasing}
\frac1\gamma \cdot c^\top x^* + \pi^\top(1-y^*_{\sfrac1\gamma})
<
\frac1y \cdot c^\top x^* + \pi^\top (1-y^*_{\sfrac1y})\enspace.
\end{align}
Now consider the function $f\colon\mathbb{R}_{\geq 1}\to\mathbb{R}$ defined by
$
\lambda \mapsto \lambda \cdot c^\top x^* + \pi^\top (1-y^*_\lambda)
$.
By definition, $f$ is convex and continuous.
Moreover, \eqref{eq:non-increasing} shows that $f$ is not increasing.
Thus, $f$ must be decreasing at the left endpoint of its domain, i.e., at $\lambda=1$, implying that there exists $\lambda_0 >1$ such that $f(\lambda_0)<f(1)$.
But $f(\lambda_0)$ is an upper bound on the objective value of the solution $(x_{\lambda_0}, y_{\lambda_0})$ of the \ref{eq:rel_PCTSP} that can be constructed through \cref{lem:PCTSP_sol_properties}, while $f(1)$ is the objective value of $(x^*, y^*)$, hence $f(\lambda_0)<f(1)$ contradicts optimality of $(x^*, y^*)$.
This finishes the proof of \cref{lem:choiceOfGamma}.
\end{proof}
\section[Proof of Theorem 10]{Proof of \cref{thm:randomizedAnalysis}}\label{sec:randomized_analysis}

\begin{proof}[Proof of \cref{thm:randomizedAnalysis}]
    The concrete choice of $b$ is left to the end of this proof; with foresight, we only require $b\in[1-e^{-\smash{\sfrac34}}, 1)$ for now.
	Let $I_b\coloneqq \int_{b}^{1} \exp(-\sfrac b\gamma)\dd\gamma$.
    Then $f(\gamma)=\sfrac{\exp(-\sfrac b\gamma)}{I_b}$.
	By \cref{thm:treeSamplingAnalysis}, the expected edge cost of the returned cycle $C=(V_C,E_C)$ is at most
	\begin{align}\label{eq:tourCost2}
		\Exp[c(E_C)]
		\le  \frac3{2}\int_b^1 \frac{f(\gamma)}{\gamma} \dd \gamma \cdot c^\top x^* =
		\frac3{2I_b}\int_b^1 \frac{\exp(-\sfrac b\gamma)}{\gamma} \dd \gamma \cdot c^\top x^*
		\enspace.
	\end{align}
	Moreover, the expected penalty cost can be bounded by
	\begin{align*}
		\Exp[\pi(V\setminus V_C)] & \le \sum_{v \in V\colon y_v^* < 1} \pi_v (1-y^*_v) \cdot \frac1{1-y_v^*}\int_{\max\{y_v^*, b\}}^{1} {\exp\left(-\frac{3y_v^*}{4\gamma}\right)}f(\gamma) \dd\gamma \\
&= \sum_{v \in V\colon y_v^* < 1} \pi_v (1-y^*_v) \cdot \underbrace{
        \frac{1}{I_b(1-y_v^*)}\cdot\int_{\max\{y_v^*, b\}}^{1} \exp\left(-\frac{3y_v^*+4b}{4\gamma}\right) \dd\gamma}_{\eqqcolon f_b(y_v^*)} \enspace.
	\end{align*}
    We investigate for which $y \in [0,1)$ the function $y\mapsto f_b(y)$ attains its maximum value.
    If $y\in[0,b]$, convexity of the map $y\mapsto g_\gamma(y)\coloneqq \frac1{1-y}\exp\big({-\frac{3y}{4\gamma}}\big)$ implies that
	$g_\gamma(y)\leq \max\{g_\gamma(0), g_{\gamma}(b)\}$.
	Since $b\geq 1-\smash{e^{-\sfrac34}}$ and because $\gamma\geq b$, the latter maximum is equal to $g_\gamma(b)$, and we can conclude that $f_b(y)$ attains its maximum for some $y\in[b,1)$, where
    $$
    f_b(y) = \frac{1}{I_b(1-y)}\cdot\int_{y}^{1} \exp\left(-\frac{3y+4b}{4\gamma}\right) \dd\gamma\enspace.
    $$
    Computational experiments suggest that $f_b(y)$ attains its maximum at $y=b$, but it seems hard to prove this analytically.
    For this reason, we aim for obtaining good bounds on the maximum value in what follows.
Observe that the integrand $h_{b,y} (\gamma) \coloneqq \exp\big({-\frac{3y+4b}{4\gamma}}\big)$ is concave over $\gamma\in[y,1]$.
    Hence,
    $$
    \int_y^1 h_{b,y}(\gamma)\dd\gamma = \int_{y}^{\frac{y+1}{2}} h_{b,y}(\gamma)\dd\gamma + \int_{\frac{y+1}{2}}^1 h_{b,y}(\gamma)\dd\gamma \leq \frac{1-y}{2} \cdot \Big(h_{b,y}\big(\textstyle\frac{3y+1}{4}\big) + h_{b,y}\big(\textstyle\frac{1+3y}{4}\big)\Big)\enspace.
    $$
    Altogether, we therefore get that
	\begin{align*}
		\max_{y \in [0,1)} f_b(y) &= \max_{y \in [b,1)} f_b(y) \\
		&= \max_{y \in [b,1)} \frac{1}{2 I_b} \Big(h_{b,y}\big(\textstyle\frac{3y+1}{4}\big) + h_{b,y}\big(\textstyle\frac{1+3y}{4}\big)\Big) \\
		&= \frac{1}{2 I_b}\cdot \max_{y \in [b,1)} \Bigg( \underbrace{\exp\left( -\frac{3y+4b}{3y+1}\right) + \exp\left( -\frac{3y+4b}{y+3}\right)}_{\eqqcolon \theta_b(y)} \Bigg) \enspace.
	\end{align*}
	Together with \eqref{eq:tourCost2}, this implies that we get in polynomial time a \PCTSP{} solution of expected value at most $\alpha\cdot\left(c^\top x^*+ \pi^\top(1-y^*)\right)$, where
	\[
	\alpha \coloneqq \frac{1}{2I_b} \cdot \max \left\{ 3\cdot\int_b^1 \frac{\exp(-\sfrac b\gamma)}{\gamma} \dd \gamma,\
	\max_{y \in [b,1)} \theta_b(y) \right\} \enspace.
	\]
    We numerically evaluate the latter expression and obtain $\alpha < 1.774$ for $b=0.6945$, as claimed in \cref{thm:randomizedAnalysis}.
    More concretely, to bound $\max_{y\in [b,1)}\theta_b(y)$, we note that the derivative of $\theta_b(y)$ on the interval $[b,1]$ can be bounded by a constant, hence the maximum can be approximated up to a prescribed error by evaluating $\theta_b(y)$ on a sufficiently fine discretization of possible values of $y$.
\end{proof}

\subsection*{Acknowledgments}

The authors are grateful to Jens Vygen for fruitful discussions, and also thank Vera Traub and Rico Zenklusen for their valuable input.
\begingroup
\setlength{\emergencystretch}{3em}
\printbibliography
\endgroup

\end{document}